\documentclass{amsart}
\usepackage{amsmath}
\usepackage{amsfonts}
\usepackage{graphics}
\usepackage{epsfig}
\usepackage{amssymb}
\usepackage{amscd}
\usepackage[all]{xy}
\usepackage{latexsym}
\usepackage{graphicx}
\usepackage{multirow}
\usepackage{geometry}

\newtheorem{thm}{Theorem}[section]

\newtheorem{lem}[thm]{Lemma}
\newtheorem{prop}[thm]{Proposition}

\newtheorem{quest}[thm]{Question}

\theoremstyle{definition}
\newtheorem{Def}[thm]{Definition}
\newtheorem{rem}[thm]{Remark}
\newtheorem*{ack}{Acknowledgement}

\newtheorem{case}{Case}

\numberwithin{equation}{section}
\numberwithin{figure}{section}

\def\End{{\text{\rm{End}}}}
\def\trace{{\text{\rm{trace}}}}

\def\rchi{{\hbox{\raise1.5pt\hbox{$\chi$}}}}
\def\Aut{{\text{\rm{Aut}}}}

\def\tensor{\otimes}

\def\Ker{{\text{\rm{Ker}}}}

\def\const{{\text{\rm{const}}}}

\def\a{\alpha}
\def\b{\beta}
\def\lam{\lambda}

\def\gam{\gamma}
\def\Gam{\Gamma}

\def\Sym{{\text{\rm{Sym}}}}

\newcommand{\bea}{\begin{eqnarray}}
\newcommand{\eea}{\end{eqnarray}}
\newcommand{\be}{\begin{equation}}
\newcommand{\ee}{\end{equation}}
\newcommand{\D}{{\frak D}}

\newcommand{\Mbar}{{\overline{\mathcal{M}}}}

\newcommand{\bP}{{\mathbb{P}}}
\newcommand{\bC}{{\mathbb{C}}}

\newcommand{\bE}{{\mathbb{E}}}

\newcommand{\bQ}{{\mathbb{Q}}}
\newcommand{\bR}{{\mathbb{R}}}

\newcommand{\bZ}{{\mathbb{Z}}}

\newcommand{\cM}{{\mathcal{M}}}

\newcommand{\cD}{{\mathcal{D}}}

\newcommand{\cH}{{\mathcal{H}}}

\newcommand{\cO}{{\mathcal{O}}}

\newcommand{\cU}{{\mathcal{U}}}

\newcommand{\la}{{\langle}}
\newcommand{\ra}{{\rangle}}
\newcommand{\half}{{\frac{1}{2}}}
\newcommand{\bp}{{\mathbf{p}}}

\newcommand{\rar}{\rightarrow}
\newcommand{\lrar}{\longrightarrow}

\newcommand{\hxi}{{\hat{\xi}}}

\begin{document}
\large
\setcounter{section}{0}

\title[Spectral curves and the Schr\"odinger equation]
{Spectral curves and the Schr\"odinger equations
for the Eynard-Orantin recursion}

\author[M.\ Mulase]{Motohico Mulase}
\address{
Department of Mathematics\\
University of California\\
Davis, CA 95616--8633, U.S.A.}
\email{mulase@math.ucdavis.edu}

\author[P.\ Su\l kowski]{Piotr Su\l kowski}
\address{
Faculty of Physics\\
 University of Warsaw\\
  ul. Ho\.za 69, 00-681 Warsaw, Poland,}
  \address{
  Institute for Theoretical Physics\\ 
  University of Amsterdam\\
Science Park 904, 1090 GL, Amsterdam, The Netherlands, and}
\address{Department of Physics\\
California Institute of Technology\\
Pasadena, CA 91125, U.S.A.}
\email{psulkows@theory.caltech.edu}

\begin{abstract}
It is predicted that the principal 
specialization of the partition function of 
a B-model topological string theory, that
is mirror dual to an A-model enumerative geometry 
problem,
satisfies a Schr\"odinger equation, and that
the characteristic variety of the 
Schr\"odinger operator gives the spectral curve
of the B-model theory, when 
an algebraic K-theory obstruction vanishes.
In this paper we present two 
concrete mathematical
A-model examples whose mirror dual partners
exhibit these predicted features on the B-model
side. The  A-model examples 
we discuss are  the
generalized  Catalan numbers
of an arbitrary genus and the  single Hurwitz numbers.
In each case,
we show that the Laplace transform of the
counting functions satisfies the
Eynard-Orantin topological recursion, that
the B-model 
partition function satisfies the KP equations,
and that the principal specialization of the partition
function satisfies a Schr\"odinger equation
whose total symbol is exactly the 
Lagrangian immersion of the spectral curve
of the Eynard-Orantin theory.

\end{abstract}

\subjclass[2000]{Primary: 14H15, 14N35, 05C30, 11P21;
Secondary: 81T30}

\begin{flushright}
\begin{tabular}{l}
\small{CALT-68-2896} \\

\\ [.3in]
\end{tabular}
\end{flushright}

\maketitle

\allowdisplaybreaks

\tableofcontents

\section{Introduction and the main results}
\label{sect:intro}

In a series of remarkable 
papers of Mari\~no \cite{M2} and  
Bouchard, Klemm, Mari\~no, and Pasquetti
\cite{BKMP},
these authors have 
developed an inductive mechanism to calculate
a variety of quantum invariants and solutions to 
enumerative geometry questions, based on the
fundamental work of
Eynard and Orantin
\cite{EO1, EO2, EO3}.
The validity of their method, known as the
\emph{remodeled B-model} based on the
\emph{topological recursion} of
Eynard-Orantin, has been established for many
different enumerative geometry problems,
such as  single Hurwitz numbers
(\cite{BEMS, EMS, MZ}, based on the conjecture
of Bouchard and Mari\~no \cite{BM}), 
open Gromov-Witten invariants of 
smooth 
toric Calabi-Yau threefolds (\cite{EO3,
Zhou2}, based on the \emph{remodeling 
conjecture} of Mari\~no \cite{M2}
and Bouchard, Klemm, Mari\~no, 
Pasquetti \cite{BKMP}), and
the number of lattice points on $\cM_{g,n}$
and its symplectic and Euclidean volumes
(\cite{CMS, DMSS, MP2012}, based on
\cite{N1,N2}). It is expected that 
double Hurwitz numbers, 
stationary Gromov-Witten invariants of
$\bP^1$ \cite{NS1, NS2}, certain Donaldson-Thomas
invariants,
and many other quantum invariants would also
fall into this category.

Unlike the familiar \emph{Topological Recursion 
Relations} (TRR) of the Gromov-Witten theory, the
Eynard-Orantin recursion is a 
B-model formula
\cite{BKMP,M2}. The significant feature of the
formula is its universality: independent of
the A-model problem, the B-model  
recursion takes always  the same form. 
The input data of this 
B-model consist of a 
holomorphic Lagrangian immersion 
\begin{equation*}
\begin{CD}
\iota :\Sigma @>>> T^*\bC
\\
&&@VV{\pi}V
\\
&&\bC
\end{CD}
\end{equation*}
of an open Riemann surface $\Sigma$ (called 
a \emph{spectral curve} of the Eynard-Orantin
recursion) into
the cotangent bundle $T^*\bC$ equipped with 
the tautological $1$-form $\eta$, and the
symmetric second derivative  of the logarithm of
\emph{Riemann's prime form} 
\cite{Fay, MumfordTataII} defined 
on $\Sigma\times \Sigma$. The procedure
of Eynard-Orantin \cite{EO1} then defines,
inductively on $2g-2+n$,
a meromorphic symmetric differential $n$-form
$W_{g,n}$ on $\Sigma^n$ for every $g\ge 0$
and $n\ge 1$ subject to $2g-2+n> 0$. 
A particular choice of the Lagrangian immersion
gives a different $W_{g,n}$, which then gives 
 a generating function of the solution to
a different enumerative geometry problem.

Thus the real question is how to find the
right Lagrangian immersion from a given
A-model. 

Suppose we have a solution to an
enumerative geometry problem (an A-model
problem). Then we
know a generating function of these quantities.
In \cite{DMSS} we proposed an idea of 
identifying the spectral curve $\Sigma$, which 
states 
that the spectral curve is the 
\emph{Laplace transform}
 of the disc amplitude of
the A-model problem. Here the Laplace transform
plays the role of  mirror symmetry. Thus we obtain 
a Riemann surface $\Sigma$. 
Still we do not see the
aspect of the Lagrangian immersion in this manner.

Every curve in $T^*\bC$ is trivially a 
Lagrangian. But not every Lagrangian is 
realized as the mirror dual to an A-model 
problem. The obstruction seems to
lie in the K-group
$K_2\big(\bC(\Sigma)\big)\tensor \bQ$. 
When this obstruction vanishes, 
we call $\Sigma$ a \emph{$K_2$-Lagrangian}, 
following Kontsevich's terminology. 
For a $K_2$-Lagrangian $\Sigma$, 
we expect the existence of a \emph{holonomic
system} that characterizes the \emph{partition
function} of the B-model theory, and at the same time,
the characteristic variety of this holonomic system
recovers the spectral curve $\Sigma$ as the
Lagrangian immersion.
A generator of this holonomic system is
called a \emph{quantum Riemann surface}
\cite{ADKMV,DHS, DHSV},
because it is a differential operator whose 
total symbol is the spectral curve
realized as a Lagrangian immersion \cite{DV2007}.
It is the work of Gukov and Su\l kowski \cite{GS}
that suggested the obstruction to the
existence of the holonomic system with 
algebraic K-theory
as an element of  $K_2$.

Another mysterious link of the 
Eynard-Orantin theory is its relation to 
integrable systems of the KP/KdV type
\cite{BE1, EO1}. We note that the
partition function of the B-model  is always
the \emph{principal specialization} of 
a $\tau$-function of the KP equations 
for all the examples we know by now.

The purpose of the present paper is to give 
 the simplest non-trivial
  mathematical examples of the
theory that exhibit these key features
mentioned above.
With these examples one can calculate 
all quantities involved, give proofs of the statements
predicted in physics, and examine  
the mathematical structure of the theory.
Our examples are based on enumeration 
problems of branched coverings of 
$\bP^1$.

The idea of \emph{homological mirror symmetry}
of Kontsevich \cite{K1994} allows us
to talk about the mirror symmetry 
\emph{without} underlying spaces,
because the formulation is based on 
the derived equivalence of
categories. Therefore, we can consider
the mirror dual B-models corresponding to 
the enumeration problems of branched 
coverings on the A-model side. 
At the same time, being the \emph{derived}
equivalence, the homological mirror 
symmetry does not tell us any direct relations 
between the quantum invariants on the 
A-model side and the complex geometry 
on the B-model side. This is exactly where
Mari\~no's idea of \emph{remodeling} 
B-model comes in
for rescue. The remodeled B-model of
\cite{BKMP, M2} is not a derived category
of coherent sheaves. Although its applicability 
is restricted to the case when there is a 
family of curves $\Sigma$ that exhibits the geometry 
of the B-model, the new idea is to construct a
network of inter-related 
 differential forms on the symmetric powers of
 $\Sigma$
 via the Eynard-Orantin
 recursion, and to understand this infinite system
 as \emph{the B-model}. The advantage of this idea
 is that we can relate the solution of the
 geometric enumeration problem on the
 A-model side and the symmetric differential
 forms on the B-model side
 through the \emph{Laplace transform}.
 In this sense we consider 
 the Laplace transform as a mirror symmetry.

The  first example 
we consider in this paper is
the generalized Catalan numbers of an arbitrary
genus. This is equivalent to the ``$c=1$ model''
of \cite[Section~5]{GS}. In terms of enumeration,
we are counting the number of algebraic curves
defined over $\overline{\bQ}$ in a systematic
way by using the dual graph of Grothendieck's
\emph{dessins d'enfants} \cite{Belyi,SL}.

Let
$D_{g,n}(\mu_1,\dots,\mu_n)$ denote the
automorphism-weighted count of the
number of connected cellular graphs 
on a closed oriented surface of 
genus $g$
(i.e., the $1$-skeleton of  cell-decompositions
of the surface),
  with $n$ labeled vertices
of degrees $(\mu_1,\dots,\mu_n)$. The letter
D stands for `dessin.'
The \emph{generalized Catalan numbers} of
type $(g,n)$ are defined by
$$
C_{g,n}(\mu_1,\dots,\mu_n) =
\mu_1\cdots\mu_n D_{g,n}(\mu_1,\dots,\mu_n).
$$
While $D_{g,n}(\vec{\mu})$ is a rational number
due to the graph automorphisms, the generalized
Catalan number
$C_{g,n}(\vec{\mu})$ is always a
non-negative  integer. It gives the 
dimension of a linear skein space. In particular,
the $(g,n)=(0,1)$ case recovers the original
Catalan numbers:
$$
C_{0,1}(2m) = C_m = \frac{1}{m+1}
\binom{2m}{m} = \dim \End_{\cU_q(s\ell_2)}
(T^{\tensor m}\bC^2).
$$
As explained in \cite{DMSS}, the 
mirror dual to the Catalan numbers $C_m$ is
the plane  curve $\Sigma$ defined by 
\begin{equation}
\label{eq:CLag}
\begin{cases}
x=z+\frac{1}{z}\\
y=-z,
\end{cases}
\end{equation}
where 
\begin{equation}
\label{eq:z(x)}
z(x) = \sum_{m=0}^\infty C_m\frac{1}{x^{2m+1}}.
\end{equation}
Note that (\ref{eq:CLag}) also gives 
a Lagrangian immersion
$\Sigma\lrar T^*\bC$. 
Let us introduce the \emph{free energies}
by
\begin{equation}
\label{eq:Cfree}
F^C_{g,n}\big(z(x_1),\dots,z(x_n)\big)
=\sum_{\vec{\mu}\in\bZ_+ ^n}
D_{g,n}(\vec{\mu}) e^{-(w_1\mu_1+\dots+
w_n\mu_n)}
=
\sum_{\vec{\mu}\in\bZ_+ ^n}
D_{g,n}(\vec{\mu}) \prod_{i=1}^n
\frac{1}{x_i ^{\mu_i}}
\end{equation}
as the Laplace transform of the number of dessins, 
where the coordinates are related by
(\ref{eq:z(x)}) and 
$x_i = e^{w_i}.$
The free energy
$F^C_{g,n}(z_1,\dots,z_n)$ is a 
symmetric function in  $n$-variables, 
and its \emph{principal specialization}
is defined by
$F^C_{g,n}(z,\dots,z)$. 
Now let 
$$
W_{g,n}^C(z_1,\dots,z_n) = d_1\cdots d_n
F_{g,n}(z_1,\dots,z_n).
$$
It is proved in \cite{DMSS} that $W_{g,n}^C$'s
satisfy the Eynard-Orantin topological 
recursion.

The \emph{Catalan partition 
function} is  given by the formula of
\cite{EO1}:
\begin{equation}
\label{eq:Cpf}
Z^C(z,\hbar)=\exp
\left(
\sum_{g=0}^\infty \sum_{n=1}^\infty
\frac{1}{n!}\; \hbar^{2g-2+n}
F^C_{g,n}(z,z,\dots,z)
\right).
\end{equation}
In this paper we prove

\begin{thm}
\label{thm:C}
The Catalan partition function satisfies the 
Schr\"odinger equation
\begin{equation}
\label{eq:Csch}
\left(
\hbar^2 \frac{d^2}{dx^2}+
\hbar x\frac{d}{dx}+1
\right) \; Z^C\big(z(x),\hbar\big)=0.
\end{equation}
The characteristic variety of this ordinary 
differential operator, $y^2+xy +1=0$
for every fixed choice of $\hbar$,  is 
exactly the 
Lagrangian immersion \rm{(\ref{eq:CLag})},
where we identify the $xy$-plane as the 
cotangent bundle $T^*\bC$ with the fiber coordinate
$y=\hbar \frac{d}{dx}$.
\end{thm}

\begin{rem}
A purely geometric reason of
 our interest in  the function appearing 
 as the principal specialization
$F^C_{g,n}(z,\dots,z)$ is that, in
 the stable range $2g-2+n>0$,  it
is a \emph{polynomial} in 
\begin{equation}
\label{eq:s}
s=\frac{z^2}{z^2-1}
\end{equation}
of degree 
$6g-6+3n$. It is indeed the virtual 
Poincar\'e polynomial of $\cM_{g,n}\times
\bR_+^n$ \cite{MP2012}, 
and its special value at $s=1$ 
gives the Euler characteristic
$(-1)^n \rchi(\cM_{g,n})$ of the moduli space
$\cM_{g,n}$ 
of smooth $n$-pointed curves of genus $g$. 
Thus $Z^C(z,\hbar)$ is the exponential 
generating function of the virtual 
Poincar\'e polynomials of $\cM_{g,n}\times
\bR_+^n$.
\end{rem}

As such, the generating function $Z^C(z,\hbar)$ is
also expressible in terms of a Hermitian matrix
integral 
\begin{equation}
\label{eq:matrix}
Z^C(z,\hbar) = \int_{\cH_{N\times N}}
\det(1-\sqrt{s}X)^N e^{-\frac{N}{2}\trace (X^2)}
dX
\end{equation}
with the identification (\ref{eq:s}) and 
$\hbar = 1/N$. Here $dX$ is the normalized
Lebesgue measure on the space of $N\times N$
Hermitian matrices $\cH_{N\times N}$.
It is a well-known fact that 
Eq.(\ref{eq:matrix}) is  the principal 
specialization of a KP $\tau$-function 
(see for example, \cite{M1994}).

Another example we consider in this
paper is based on 
single Hurwitz numbers. As a counting problem
it is easier to state than the previous example,
but the Lagrangian immersion requires a
transcendental function, and hence the
resulting Schr\"odinger equation exhibits a
rather different nature. 

Let $H_{g,n}(\mu_1\dots,\mu_n)$ be the
automorphism-weighted count of the 
number of topological types of  
Hurwitz covers $f:C\lrar \bP^1$ of
a connected non-singular algebraic curve 
$C$ of genus $g$. A holomorphic map $f$
is a \emph{Hurwitz cover} of 
\emph{profile} $(\mu_1\dots,\mu_n)$
if it has $n$ labeled 
poles of orders $(\mu_1\dots,\mu_n)$
and is simply ramified otherwise. Introduce 
the Laplace transform of single Hurwitz 
numbers by 
\begin{equation}
\label{eq:FgnH intro}
F_{g,n}^H\big(t(w_1),\dots,t(w_n)\big)
= \sum_{\vec{\mu}\in\bZ_+ ^n}
H_{g,n}(\vec{\mu}) e^{-(w_1\mu_1+\cdots
w_n\mu_n)},
\end{equation}
where 
$$
t(w) = \sum_{m=0}^\infty \frac{m^m}{m!}e^{-mw}
$$
is the tree-function. Here again $F_{g,n}^H(t_1,
\dots,t_n)$ is a polynomial of degree $6g-6+3n$
if $2g-2+n>0$
\cite{MZ}. Bouchard and Mari\~no
have conjectured \cite{BM} that 
$$
W_{g,n}^H(t_1,\dots,t_n) = d_1\cdots d_n
F_{g,n}^H(t_1,\dots,t_n)
$$
satisfy the Eynard-Orantin topological recursion,
with respect to the Lagrangian immersion 
\begin{equation}
\label{eq:Lambert intro}
\begin{cases}
x = e^{-w} = z e^{-z}\in \bC^*\\
y = z\in \bC,
\end{cases}
\end{equation}
where 
$$
z = \frac{t-1}{t}
$$
and we use $\eta = y\frac{dx}{x}$ as the tautological
$1$-form on $T^*\bC^*$. 
The Bouchard-Mari\~no
conjecture was proved in 
\cite{BEMS, EMS, MZ}.

Now we define the \emph{Hurwitz partition
function}
\begin{equation}
\label{eq:ZH intro}
Z^H(t,\hbar)=
\exp\left(\sum_{g=0}^\infty \sum_{n=1}^\infty
\frac{1}{n!}\;\hbar^{2g-2+n}F_{g,n}^H(t,\dots,t)
\right).
\end{equation}
Then we have the following

\begin{thm}
The Hurwitz partition function satisfy 
 two equations:
\begin{align}
\label{eq:Sch H1 intro}
&\left[\half \hbar \frac{\partial^2}{\partial w^2}
+\left(1+ \half{\hbar}\right)
\frac{\partial}{\partial w}
-\hbar \frac{\partial}{\partial \hbar}
\right] Z^H\big(t(w),\hbar\big) = 0,
\\
\label{eq:Sch H2 intro}
&
\left(
\hbar \frac{d}{dw}+e^{-w}e^{-\hbar \frac{d}{dw}}
\right)Z^H\big(t(w),\hbar\big) = 0.
\end{align}
Moreover, each of the two equations recover
the Lagrangian immersion 
{\rm{(\ref{eq:Lambert intro})}} from the 
asymptotic analysis at $\hbar \sim 0$. 
And if we view $\hbar$ as a fixed constant in 
{\rm{(\ref{eq:Sch H2 intro})}}, then its
total symbol is the Lagrangian immersion 
{\rm{(\ref{eq:Lambert intro})}} with the 
identification $z = -\hbar \frac{d}{dw}$.
\end{thm}

\begin{rem}
The second order equation
(\ref{eq:Sch H1 intro}) is a consequence
of the polynomial recursion of \cite{MZ}.
This situation is exactly the same as
Theorem~\ref{thm:C}.
The differential-difference equation
(\ref{eq:Sch H2 intro}), or
a \emph{delay} differential equation,
 follows from the
principal specialization of the 
KP $\tau$-function that
gives another generating function of
single Hurwitz numbers 
\cite{Kazarian, O}. We remark that
(\ref{eq:Sch H2 intro}) is also derived in
\cite{Zhou4}.
The point of view of differential-difference
equation is further developed in
\cite{MSS} for the case of double Hurwitz
numbers and $r$-spin structures, where we
generalize a result of \cite{Zhou4}.
\end{rem}

\begin{rem}
Define two operators by
\begin{align}
\label{eq:P}
P&=
\hbar \frac{d}{dw}+e^{-w}e^{-\hbar \frac{d}{dw}}
\qquad {\text{and}}
\\
Q&=
\half \hbar \frac{\partial^2}{\partial w^2}
+\left(1+ \half{\hbar}\right)
\frac{\partial}{\partial w}
-\hbar \frac{\partial}{\partial \hbar}.
\label{eq:Q}
\end{align}
Then it is noted in \cite{LMS} that
\begin{equation}
\label{eq:PQcommutator}
[P,Q]=P.
\end{equation}
Thus the \emph{heat equation} (\ref{eq:Sch H1 intro})
preserves the space of solutions of the
\emph{Schr\"odinger equation} 
(\ref{eq:Sch H2 intro}). In this sense, 
(\ref{eq:Sch H2 intro}) is holonomic for every fixed 
$\hbar$. The analysis of these equations is further
investigated in \cite{LMS}.
\end{rem}

The existence of a holonomic
system is particularly appealing when we consider
the knot A-polynomial as the defining
equation of a Lagrangian immersion, 
in connection to the AJ conjecture
\cite{Gar, GarLe,Gukov}.
One can ask:

\begin{quest}
Let $K$ be a knot in $S^3$ and $A_K$  its 
A-polynomial \cite{CCGLS}. Is there a concrete
formula for the quantum knot invariants 
of $K$, such as the colored Jones polynomials,
in terms of $A_K$?
\end{quest}

The B-model developed in \cite{DFM}, 
\cite{GS}, and more recently in \cite{BE2},
clearly shows  that the answer is yes, and
it should be given by the 
Eynard-Orantin formalism. 
Although our examples in the
present paper are not related to any knots,
they suggest the existence of  a corresponding
A-model. An interesting theory
of generalized A-polynomials and
quantum knot invariants from the point of view
of mirror symmetry is recently presented in
 Aganagic and Vafa \cite{AV}. 
We also remark that there are
further developments in this 
direction \cite{BEM, FGS1, FGS2}.

The paper is organized as follows.
In Section~\ref{sect:EO}, we give 
the definition of the Eynard-Orantin topological
recursion. We emphasize the aspect of
Lagrangian
immersion in our presentation.
In Section~\ref{sect:Catalan} 
we review the generalized Catalan numbers
of \cite{DMSS}. Then in Section~\ref{sect:CSch}
we derive the Schr\"odinger equation 
for the Catalan partition function.
The equation for the Hurwitz partition function
is given in Section~\ref{sect:Hurwitz}. 
Finally, in Section~\ref{sect:Schur}
we give the proof of (\ref{eq:Sch H2 intro}) 
using the Schur function expansion of the
Hurwitz generating function and its 
principal specialization.

\section{The Eynard-Orantin  topological recursion}
\label{sect:EO}

We adopt the following definition 
for the topological recursion of Eynard-Orantin
 \cite{EO1}. Our emphasis, which is different
 from the original, is the point of view of 
 the Lagrangian immersion.

\begin{Def}
\label{def:spectral}
The \textbf{spectral curve} $(\Sigma,\iota)$
consists of 
 an open Riemann surface $\Sigma$ and 
 a \emph{Lagrangian immersion}
\begin{equation}
\label{eq:LI}
\begin{CD}
\iota :\Sigma @>>> T^*\bC
\\
&&@VV{\pi}V
\\
&&\bC
\end{CD}
\end{equation}
with respect to the 
canonical holomorphic 
symplectic structure $\omega = -d\eta$ on 
$T^*\bC$,
where $\eta$ is the tautological $1$-form on 
the cotangent bundle $\pi:T^*\bC\lrar \bC$.
Recall that $p\in\Sigma$ is a \emph{Lagrangian
singularity} if $d(\pi\circ\iota)=0$ at $p$, and that
$\pi(\iota(p))\in \bC$ is a \emph{caustic} 
of the Lagrangian
immersion. We assume that the projection $\pi$
restricted to the Lagrangian immersion is 
locally simply ramified around each 
Lagrangian singularity.
We denote by $R=\{p_1,\dots,p_r\}\subset 
\Sigma$ 
 the set of Lagrangian singularities, and by 
$$
U = \sqcup_{j=1} ^r U_j
$$
the disjoint union of small neighborhood
 $U_j$ around $p_j$ such that 
 $\pi:U_j\lrar \pi(U_j)\subset \bC$ is
 a double-sheeted covering ramified only at $p_j$.
 We denote by $s_j(z)$ the local 
 Galois conjugate
of $z\in U_j$.
\end{Def}

Another key ingredient of the Eynard-Orantin
theory is the \textbf{normalized
fundamental differential
of the second kind} $B_\Sigma(z_1,z_2)$
\cite[Page 20]{Fay}, \cite[Page 3.213]{MumfordTataII}, which is a symmetric
differential $2$-form on $\Sigma\times \Sigma$
with second-order poles along the diagonal. 
To define it, let us recall a few basic facts about
algebraic curves. Let $C$ be a nonsingular complete
algebraic curve over $\bC$. 
We are considering a nonsingular compactification
$C=\overline{\Sigma}$ of 
the Riemann surface $\Sigma$. 
We identify the Jacobian
variety of $C$ as 
$Jac(C) = Pic^0(C)$, which is isomorphic to
$Pic^{g-1}(C)$. The \emph{theta divisor}
 $\Theta$ of 
$Pic^{g-1}(C)$ is defined by 
$$
\Theta = \{L\in Pic^{g-1}(C)\;|\; \dim H^1(C,L)>0\}.
$$
We use the same notation for the translate divisor
on $Jac(C)$, also called the theta divisor. 
Now consider the diagram
\begin{equation*}
		\xymatrix{& Jac(C)
		\\
		&C\times C\ar[dl]_{pr_1} \ar[u]_{\delta}
		\ar[dr]^{pr_2}
		\\
		C& & 	C ,
		}
\end{equation*} 
where $pr_j$ denotes the projection to the
$j$-th components,
and 
$$
\delta:C\times C \owns (p,q)\longmapsto p-q\in
Jac(C).
$$
Then the \emph{prime form} $E_C(z_1,z_2)$ 
\cite[Page 16]{Fay} is
defined as a holomorphic section 
$$
E_C(p,q) \in H^0\left(C\times C,
pr_1^* K_C^{-\half}\tensor 
pr_2^* K_C^{-\half}\tensor
\delta^*(\Theta)
\right),
$$
where $K_C$ is the canonical line bundle of 
$C$ and we choose  Riemann's spin structure
(or the Szeg\"o kernel) $K_C^\half$ 
(see \cite[Theorem~1.1]{Fay}).
We do not need the precise definition of the 
prime form here, but its characteristic properties
are important:
\begin{enumerate}
\item $E_C(p,q)$ vanishes only along the
diagonal $\Delta\subset C\times C$, and has
simple zeros along $\Delta$.
\item Let $z$ be a local coordinate on 
$C$ such that $dz(p)$ gives the local trivialization of
$K_C$ around $p$. When $q$ is near at $p$,
 $\delta^*(\Theta)$ is also trivialized 
around $(p,q)\in C\times C$, and we have a local
behavior
\begin{equation}
\label{eq:E}
E_C\big(z(p),z(q)\big) = 
\frac{z(p)-z(q)}{\sqrt{dz(p)}\cdot
\sqrt{dz(q)}}\left(1+O\big((z(p)-z(q))^2\big)
\right).
\end{equation}
\item $E_C\big(z(p),z(q)\big) 
= -E_C\big(z(q),z(p)\big) $.
\end{enumerate}

The fundamental $2$-form $B_C(p,q)$
is then defined by
\begin{equation}
\label{eq:B}
B_C(p,q) = d_1\tensor d_2 \log E_C(p,q)
\end{equation}
(aee \cite[Page 20]{Fay}, 
\cite[Page 3.213]{MumfordTataII}).
We note that $dz(p)$ appears in (\ref{eq:E}) just
as the indicator of our choice of the
local trivialization. 
With this local trivialization, the square
$$
E_C(p,q)^2 \in H^0\left(C\times C,
pr_1^* K_C^{-1}\tensor 
pr_2^* K_C^{-1}\tensor
\delta^*(\Theta)^{\tensor 2}
\right)
$$
behaves better because of
\begin{equation}
\label{eq:E2}
E_C\big(z(p),z(q)\big)^2 = 
\frac{\big(z(p)-z(q)\big)^2}{{dz(p)}\cdot
{dz(q)}}\left(1+O\big((z(p)-z(q))^2\big)
\right).
\end{equation}
We then see that
\begin{multline}
\label{eq:B local}
B_C\big(z(p),z(q)\big) =\half
d_1\tensor d_2 \log E^2\big(z(p),z(q)\big)
\\
=
\frac{dz(p)\cdot dz(q)}{\big(z(p)-z(q)\big)^2}
+O(1)\;dz(p)\cdot dz(q)
\\
\in H^0\left(C\times C,pr_1^*K_C\tensor
pr_2^* K_C\tensor \cO(2\Delta)
\right).
\end{multline}

\begin{Def}
\label{def:symmetric}
Let $D$ be a divisor of $\Sigma$. 
A \textbf{meromorphic symmetric differential
form} of degree $n$
with poles along $D$ is an element of 
the symmetric tensor product
$$
\Sym^n H^0 \left( \Sigma,K_\Sigma (*D)
\right).
$$
\end{Def}

\begin{Def}
\label{def:EO}
Meromorphic differential
forms 
$$
W_{g,n}(z_1,\dots,z_n)
\in \Sym^n H^0\left(
\Sigma,K_\Sigma(*R)\right)
$$
 for $g\ge 0$ and $n\ge 1$, subject to
$2g-2+n>0$, with poles along the
Lagrangian singularities of the 
Lagrangian immersion $\Sigma\lrar T^*\bC$,
are said to satisfy the \textbf{Eynard-Orantin
topological recursion}
 if they satisfy the recursion formula
 \begin{multline}
\label{eq:EO}
W_{g,n}(z_1,z_2,\dots,z_n)
= \frac{1}{2\pi i} 
\sum_{j=1} ^r 
\oint_{\gam_j} K_j(z,z_1)
\Bigg[
W_{g-1,n+1}\big(z,s_j(z),z_2,\dots,z_n\big)\\
+
\sum^{\text{No $(0,1)$ terns}} _
{\substack{g_1+g_2=g\\I\sqcup J=\{2,3,\dots,n\}}}
W_{g_1,|I|+1}(z,z_I) W_{g_2,|J|+1}(s_j(z),z_J)
\Bigg].
\end{multline}
Here the integration is taken with respect to
$z\in U_j$ along a positively oriented 
simple closed loop $\gam_j$
around $p_j$, and
$z_I = (z_i)_{i\in I}$
for a subset $I\subset \{1,2,\dots,n\}$.
In the summation, ``No $(0,1)$ terms'' means
the summand does not
contain the terms with $g_1=0$ and $I=\emptyset$ or
$g_2=0$ and $J=\emptyset$.
 The recursion kernels $K_j(z,z_1)$, $j=1, \dots, r$,
 are defined as follows. 
First we define $W_{0,1}$
and $W_{0,2}$.
\begin{equation}
\label{eq:W01}
W_{0,1}(z)=\iota^*\eta =ydx
\in H^0(\Sigma,K),
\end{equation}
where $x$ is a linear coordinate
of $\bC$ and $y$ is the dual coordinate 
of $T^*_0 \bC$ so that the Lagrangian immersion is
 given by 
$$
(x,y):\Sigma\owns t\longmapsto (x(z),y(z))\in 
T^*\bC.
$$
\begin{equation}
\label{eq:W02}
W_{0,2}(z_1,z_2) = 
B_\Sigma(z_1,z_2)
 -\pi^*
\frac{dx_1\cdot dx_2}
{(x_1-x_2)^2}. 
\end{equation}
We note that 
$W_{0,2}(z_1,z_2)$ is holomorphic
along the diagonal $z_1=z_2$.
The recursion kernel $K_j(z_1,z_2)
\in H^0
\left(U_j\times \Sigma,K_{U_j}^{-1}\tensor
K_\Sigma \right)$
for $z_1\in U_j$ and $z_2\in \Sigma$ is defined by
\begin{multline}
\label{eq:kernel}
K_j(z_1,z_2) =
\half\; \frac{1}
{W_{0,1}\big(s_j(z_1)\big)-W_{0,1}(z_1)}
\tensor \int_{z_1} ^{s_j(z_1)} W_{0,2}
(\;\cdot\;,z_2)
\\
=
\half\;
\frac{d_2\left(\log E_\Sigma\big(s_j(z_1), z_2\big)
-\log E_\Sigma(z_1, z_2)\right)}
{\left(y(s_j(z_1)\big)-y(z_1)\right)dx(z_1)}.
\end{multline}
The  kernel $K_j(z_1,z_2)$ is an algebraic operator
that multiplies $dz_2$ while contracts
$\partial/\partial z_1$.
\end{Def}

The topological 
recursion is the reduction of 
$2g-2+n$ by $1$, which is different
from the usual boundary degeneration formula
of $\Mbar_{g,n}$. As shown in 
Figure~\ref{fig:recursion},
the reduction corresponds to degeneration cycles
of codimension $1$ and $2$,
as in Arbarello-Cornalba-Griffiths
\cite[Chapter 17, Section 5, Page 589]{ACG}.
\begin{align*}
&\cM_{0.3}\times \Mbar_{g,n-1}
\lrar \Mbar_{g,n}, 
\qquad
\cM_{0.3}\times \Mbar_{g-1,n+1}
\lrar \Mbar_{g,n},
\\
&\cM_{0.3}\times 
\bigcup_{\substack{g_1+g_2=g\\
n_1+n_2=n-1}}
\Mbar_{g_1,n_1+1}
\times \Mbar_{g_2,n_2+1}
\lrar \Mbar_{g,n}
\end{align*}

\begin{figure}[htb]
\centerline{\epsfig{file=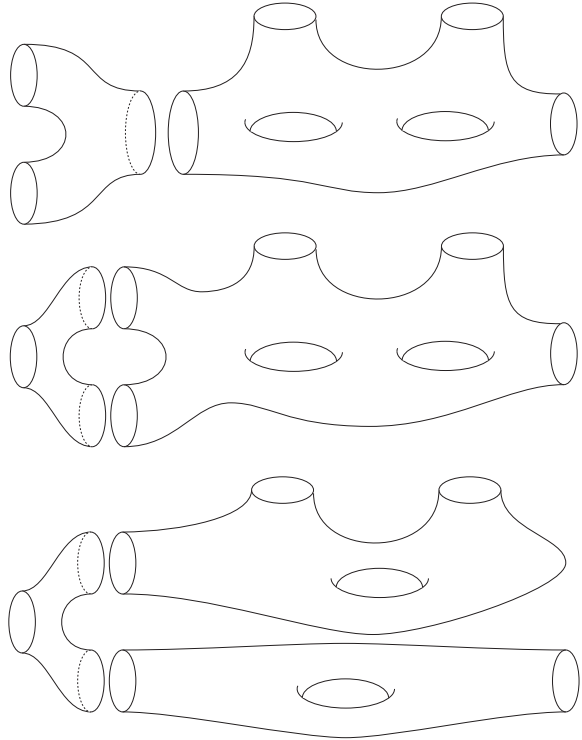, 
width=3in}}
\caption{The topological recursion 
and degeneration.}
\label{fig:recursion}
\end{figure}

The recursion starts from $W_{1,1}$ and
$W_{0,3}$. If we modify (\ref{eq:EO})
slightly, then these can also be calculated
from $W_{0,2}$ \cite{EO1}.

\begin{equation}
\label{eq:W11}
W_{1,1}(z_1) = 
 \frac{1}{2\pi i} 
\sum_{j=1} ^r 
\oint_{\gam_j} K_j(z,z_1) 
\left.\left[
W_{0,2}(u,v)+c\pi^* \frac{dx(u)\cdot dx(v)}
{(x(u)-x(v))^2}
\right]
\right|_{u=z, v=s_j(z)}.
\end{equation}
Formula (\ref{eq:EO})
does not give an apparently symmetric
expression for $W_{0,3}$. 
In terms of the coordinate
$(x,y)\in T^*\bC$ we have an alternative
formula for $W_{0,3}$ \cite{EO1}:
\begin{equation}
\label{eq:W03}
W_{0,3}(z_1,z_2,z_3) =\frac{1}{2\pi i} 
\sum_{j=1} ^r 
\oint_{\gam_j}
\frac{W_{0,2}(z,z_1)W_{0,2}(z,z_2)
W_{0,2}(z,z_3)}
{dx(z)\cdot dy(z)}.
\end{equation}

Suppose we have a solution $W_{g,n}$
to the 
topological recursion. A \textbf{primitive functions}
of the symmetric
differential form $W_{g,n}$ is a 
symmetric meromorphic function $F_{g,n}$ on 
$\Sigma^n$ such that its $n$-fold exterior
derivative recovers the $W_{g,n}$, i.e.,
\begin{equation}
\label{eq:W=dF}
W_{g,n}(z_1,\dots,z_n) = 
d_{z_1}\cdots d_{z_n} F_{g,n}(z_1,\dots,z_n).
\end{equation}
The \textbf{partition function}
of the topological recursion 
for a genus $0$ spectral curve 
is the formal expression
in infinitely many variables
\begin{equation}
\label{eq:partition function}
Z(z_1,z_2,\dots;\hbar)
= \exp\left(
\sum_{g\ge 0,n\ge 1}\frac{1}{n!}\;\hbar^{2g-2+n}
F_{g,n}(z_1,z_2,\dots,z_n)
\right).
\end{equation}
The \emph{principal specialization}
of the partition function is also denoted by 
the same letter $Z$:
\begin{equation}
\label{eq:principal}
Z(z,\hbar)
= \exp\left(
\sum_{g\ge 0,n\ge 1}\frac{1}{n!}\;\hbar^{2g-2+n}
F_{g,n}(z,z,\dots,z)
\right).
\end{equation}

\begin{rem}
The partition function 
depends on the choice of the primitive
functions.
When we consider the topological recursion 
as the B-model corresponding to an A-model
counting problem, then there is always
a canonical choice for the primitives, as the 
Laplace transform of the quantum invariants.
\end{rem}

\begin{rem}
When the spectral curve $\Sigma$ has a higher
genus, the partition function requires a 
non-perturbative factor in terms of a theta function
associated to the curve \cite{BE1,BE2}. 
In this case the algebraic K-theory condition
of \cite{GS},
probably similar to the \emph{Boutroux condition}
of \cite{ADKMV}, becomes essential for the
existence of the quantum curve or the Schr\"odinger
equation. 
Our consideration in
paper is limited to the genus $0$ case.
\end{rem}

\section{The generalized 
Catalan numbers and the 
topological recursion}
\label{sect:Catalan}

A \textbf{cellular graph} of type $(g,n)$
is the one-skeleton of a cell-decomposition of 
a  connected closed oriented surface
of genus $g$ with $n$ 
$0$-cells labeled by the index set 
$[n]=\{1,2,\dots,n\}$. Two cellular graphs
are identified if an orientation-preserving 
homeomorphism  of a surface into another
surface maps one cellular graph to another,
honoring the labeling of each vertex.
Let $D_{g,n}(\mu_1,\dots, \mu_n)$ denote
the number of connected cellular graphs $\Gam$ of
type $(g,n)$ with $n$ labeled vertices
of degrees $(\mu_1,\dots,\mu_n)$, 
counted with the weight 
$1/|\Aut(\Gam)|$. It is generally a rational number.
The orientation of the surface induces a 
cyclic order of incident half-edges at each 
vertex of a cellular graph $\Gam$. Since 
$\Aut(\Gam)$ fixes each vertex, the automorphism
group is a subgroup of the Abelian group
$\prod_{i=1} ^n \bZ\big/\mu_i \bZ$ that rotates
each vertex.
Therefore, 
\begin{equation}
\label{eq:Catalan gn}
C_{g,n}(\mu_1,\dots,\mu_n)
= \mu_1\cdots\mu_n D_{g,n}(\mu_1,\dots,\mu_n)
\end{equation}
is always an integer. 
The cellular graphs counted by
(\ref{eq:Catalan gn}) are connected 
 graphs of genus $g$ with $n$ vertices of degrees
$(\mu_1,\dots,\mu_n)$, and at the $j$-th vertex 
for every $j=1,\dots,n$, an outgoing
arrow is placed on one of the incident
$\mu_j$ half-edges (see Figure~\ref{fig:cellulargraph}).
The placement of $n$ arrows corresponds to the
factors $\mu_1\cdots\mu_n$ on the right-hand side.
We call this integer the generalized
\textbf{Catalan number} of type $(g,n)$.
The reason for this naming comes from the
following theorem.

\begin{figure}[htb]
\centerline{
\epsfig{file=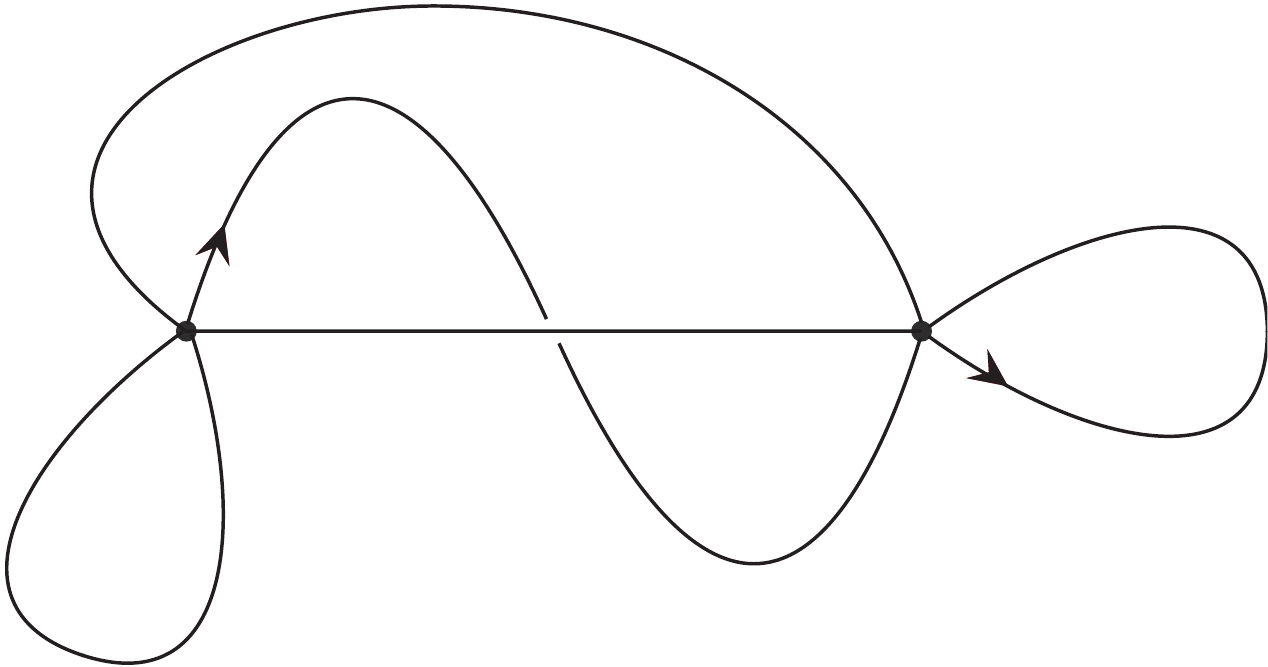, width=2in}}
\caption{A cellular graph of type $(1,2)$.}
\label{fig:cellulargraph}
\end{figure}

\begin{thm}
The generalized Catalan numbers 
{\rm(\ref{eq:Catalan gn})}
satisfy the following equation.
\begin{multline}
\label{eq:Catalan recursion}
C_{g,n}(\mu_1,\dots,\mu_n) 
=\sum_{j=2}^n \mu_j 
C_{g,n-1}(\mu_1+\mu_j-2,\mu_2,\dots,
\widehat{\mu_j},\dots,\mu_n)
\\
+
\sum_{\a+\b = \mu_1-2}
\left[
C_{g-1,n+1}(\a,\b,\mu_2,\cdots,\mu_n)+
\sum_{\substack{g_1+g_2=g\\
I\sqcup J=\{2,\dots,n\}}}
C_{g_1,|I|+1}(\a,\mu_I)C_{g_2,|J|+1}(\b,\mu_J)
\right],
\end{multline}
where $\mu_I=(\mu_i)_{i\in I}$ for 
an index set $I\subset[n]$,
$|I|$ denotes the cardinality of $I$, and
the third sum in the formula is for 
all  partitions of $g$ and 
set partitions of $\{2,\dots,n\}$. 
\end{thm}

\begin{proof}
Let $\Gam$ be an arrowed cellular graph counted
by the left-hand side of (\ref{eq:Catalan recursion}).
Since all  vertices of $\Gam$ are labeled, let
$\{p_1,\dots,p_n\}$ denote the vertex set.
We look at the half-edge incident to $p_1$ that
carries an arrow. 

\begin{case}
The arrowed half-edge  extends to an edge $E$ that
connects $p_1$ and $p_j$ for some $j>1$.
\end{case}

In this case, we shrink the edge and join the two
vertices $p_1$ and $p_j$ together. By this process
we create a new vertex of degree $\mu_1+\mu_j-2$.
To make the counting bijective, we need to be able
to go back from the shrunken graph to the original,
provided that we know $\mu_1$ and $\mu_j$.
Thus we place an arrow to the half-edge 
next to $E$ around $p_1$ with respect to the
counter-clockwise cyclic order that comes from 
the orientation of the surface. In this process
we have $\mu_j$ different arrowed graphs 
that produce the same result, because we must
remove the arrow placed around the vertex $p_j$
in the original graph. This gives the 
right-hand side of the first line of 
(\ref{eq:Catalan recursion}). See 
Figure~\ref{fig:case1}.

\begin{figure}[htb]
\centerline{
\epsfig{file=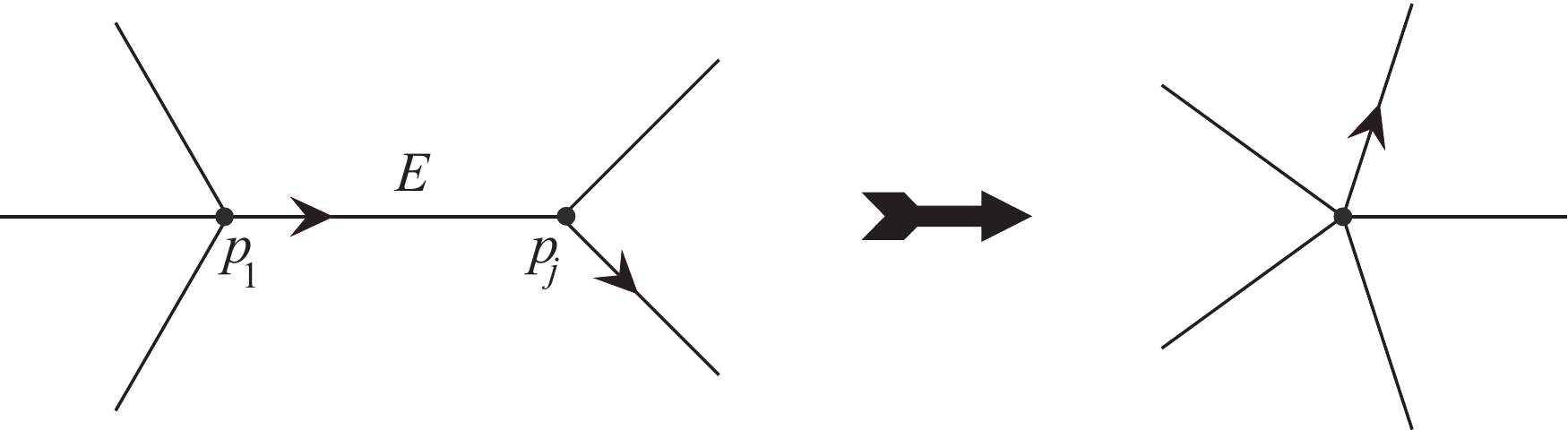, width=2.5in}}
\caption{The process of shrinking the arrowed edge
$E$ that connects vertices $p_1$ and $p_j$, $j>1$.}
\label{fig:case1}
\end{figure}

\begin{case}
The arrowed half-edge at $p_1$ is  a loop
$E$ that goes out and comes back to $p_1$.
\end{case}

The process we apply is again shrinking the loop $E$.
The loop $E$ separates all other half-edges into 
two groups, one consisting of $\a$ of them placed on
one side of the loop, and the other consisting of 
$\b$ half-edges placed on the other side. It can 
happen that $\a=0$ or $\b=0$. 
Shrinking a loop on a surface causes pinching. 
Instead of creating a pinched (i.e., singular) surface, 
we separate the double point into two new vertices
of degrees $\a$ and $\b$. 
Here again we need to remember the place of the 
loop $E$. Thus we place an arrow to the half-edge
next to the loop in each group.
See Figure~\ref{fig:case2}.

\begin{figure}[htb]
\centerline{
\epsfig{file=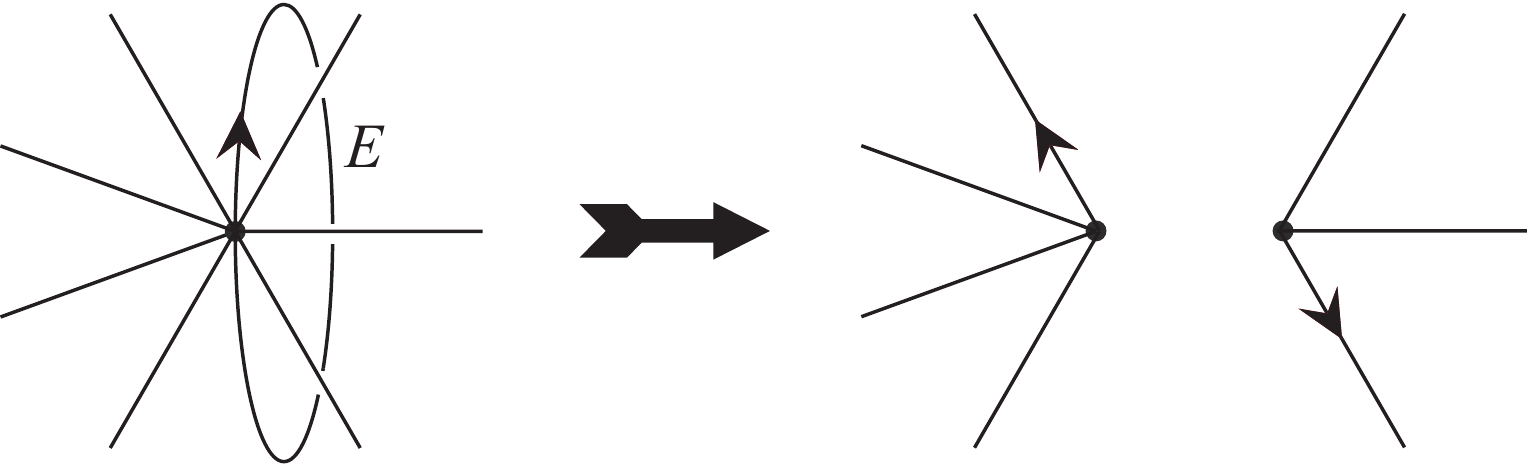, width=2.5in}}
\caption{The process of shrinking the arrowed loop
$E$ that is attached to  $p_1$.}
\label{fig:case2}
\end{figure}

After the pinching and separating the double 
point, the original surface of genus $g$ with 
$n$ vertices $\{p_1,\dots,p_n\}$ may change its
topology. It may have genus $g-1$, or it splits into
two pieces of genus $g_1$ and $g_2$. 
The second line of (\ref{eq:Catalan recursion})
records all such cases. This completes the proof.
\end{proof}

\begin{rem}
For $(g,n) = (0,1)$, the above formula reduces to 
\begin{equation}
\label{eq:Cm}
C_{0,1}(\mu_1) =
\sum_{\a+\b=\mu_1-2} C_{0,1}(\a) C_{0,1}(\b).
\end{equation}
When $n=1$, the degree of the unique vertex $\mu_1$
is always even. By defining $C_{0,1}(0)=1$, 
we find that 
$$
C_{0,1}(2m)=C_m=\frac{1}{m+1}\binom{2m}{m}
$$
is the $m$-th Catalan number. Only for $(g,n)=(0,1)$
we have this irregular case of $\mu_1=0$ 
happens, because a degree $0$ single vertex is
\emph{connected}, and gives a cell-decomposition
of $S^2$. We can imagine that a single vertex on
$S^2$ has
an infinite cyclic group as its automorphism, 
so that $C_{0,1}(0)=1$ is consistent
with 
$$
C_{0,1}(\mu_1) = \mu_1 D_{0,1}(\mu_1).
$$ 
All other cases,
if one of the verteces has degree $0$, then the 
Catalan number $C_{g,n}$ 
is simply $0$ because of the
definition
(\ref{eq:Catalan gn}).
\end{rem}

Let us introduce the generating function of the 
Catalan numbers by
\begin{equation}
\label{eq:Catalan z}
z=z(x) = \sum_{m=0} ^\infty \frac{1}{x^{2m+1}}.
\end{equation}
Then by the quadratic recursion (\ref{eq:Cm})
we find that the inverse function of $z(x)$
that vanishes at $x=\infty$ is given by
\begin{equation}
\label{eq:Catalan x}
x=z+\frac{1}{z}.
\end{equation}
This defines a Lagrangian immersion
\begin{equation}
\label{eq:Catalan immersion}
\Sigma=\bC \owns z\longmapsto
(x(z),y(z))\in T^*\bC,
\qquad 
\begin{cases}
x(z) = z+\frac{1}{z}\\
y(z) = -z
\end{cases}.
\end{equation}
The Lagrangian singularities are located 
at the points at which $dx=0$, i.e., $z=\pm 1$. 
Often it is more convenient to use the 
coordinate 
\begin{equation}
\label{eq:Catalan t}
z = \frac{t+1}{t-1}.
\end{equation}
The following theorem is established in 
\cite{DMSS}.

\begin{thm}[\cite{DMSS}]
The Laplace transform of the Catalan numbers of
type $(g,n)$ defined as symmetric differential forms by
\begin{equation}
\label{eq:WgnC}
W_{g,n}^C(t_1,\dots,t_n) =
(-1)^n
\sum_{(\mu_1,\dots,\mu_n)\in\bZ_+ ^n}
C_{g,n}(\mu_1,\dots,\mu_n) \;e^{-\la w,\mu\ra}
dw_1\cdots dw_n
\end{equation}
satisfies the Eynard-Orantin recursion with respect
to the Lagrangian immersion
 {\rm(\ref{eq:Catalan immersion})}
and 
\begin{equation}
\label{eq:W02C}
W_{0,2}^C(t_1,t_2) = \frac{dt_1\cdot dt_2}
{(t_1-t_2)^2}
-\frac{dx_1\cdot dx_2}{(x_1-x_2)^2}.
\end{equation}
Here the Laplace transform coordinate $w$ is
related to the coordinate $t$ of the Lagrangian by
$$
e^{w_i} = x_i = z_i+\frac{1}{z_i} =
 \frac{t_i+1}{t_i-1}
+\frac{t_i-1}{t_i+1}, \qquad i=1,2,\dots,n,
$$
and $\la w,\mu\ra = w_1\mu_1+\cdots+w_n\mu_n$.
\end{thm}

In this case the Eynard-Orantin
 recursion formula is given by
\begin{multline}
\label{eq:CEO}
W_{g,n}^C(t_1,\dots,t_n)
=
-\frac{1}{64} \; 
\frac{1}{2\pi i}\int_\gam
\left(
\frac{1}{t+t_1}+\frac{1}{t-t_1}
\right)
\frac{(t^2-1)^3}{t^2}\cdot \frac{1}{dt}\cdot dt_1
\\
\times
\Bigg[
\sum_{j=2}^n
\bigg(
W_{0,2}^C(t,t_j)W_{g,n-1}^C
(-t,t_2,\dots,\widehat{t_j},
\dots,t_n)
+
W_{0,2}^D(-t,t_j)W_{g,n-1}^C
(t,t_2,\dots,\widehat{t_j},
\dots,t_n)
\bigg)
\\
+
W_{g-1,n+1}^C(t,{-t},t_2,\dots,t_n)
+
\sum^{\text{stable}} _
{\substack{g_1+g_2=g\\I\sqcup J=\{2,3,\dots,n\}}}
W_{g_1,|I|+1}^C(t,t_I) W_{g_2,|J|+1}^C({-t},t_J)
\Bigg].
\end{multline}
The last sum is restricted to the stable 
geometries. In other words, the partition 
should satisfies
$2g_1-1+|I|>0$ and $2g_2-1+|J|$. The 
contour integral is taken with respect to 
$t$ on the curve defined by Figure~\ref{fig:contourC}.

\begin{figure}[htb]
\centerline{\epsfig{file=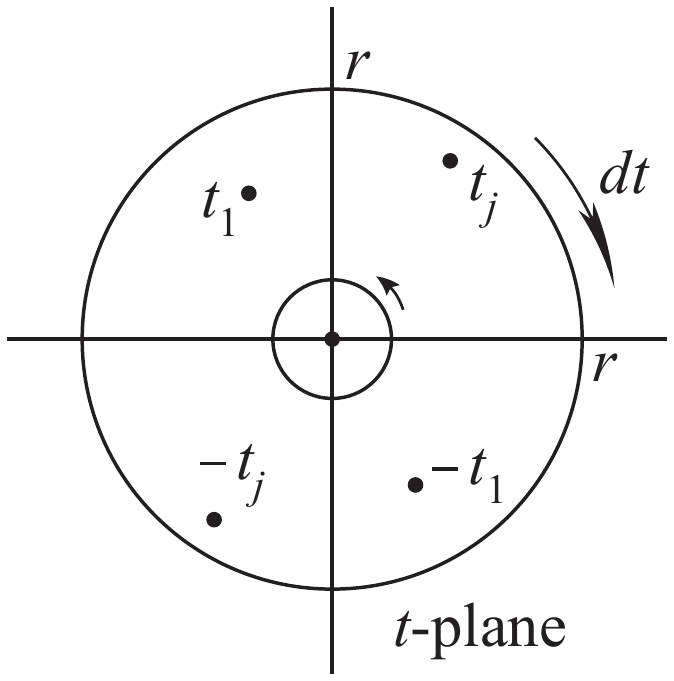, width=1.5in}}
\caption{The integration contour $\gamma$. This contour 
encloses an annulus bounded by two concentric 
circles centered at the origin. The outer one has a 
large radius 
$r>\max_{j\in N} |t_j|$ and the negative orientation,
 and the inner one has an infinitesimally small radius with 
 the positive
 orientation.}
\label{fig:contourC}
\end{figure}

\begin{rem}
The recursion formula (\ref{eq:CEO}) is 
a genuine induction formula with respect to 
$2g-2+n$. Thus from the given $W_{0,1}^C$
and $W_{0,2}^C$, we can calculate 
all $W_{g,n}^C$ one by one. This is a big
difference between (\ref{eq:CEO})
and (\ref{eq:Catalan recursion}). The latter
relation contains the terms with $C_{g,n}$ in the 
right-hand side as well.
\end{rem}

\section{The partition function for the
generalized Catalan numbers and the
Schr\"odinger equation}
\label{sect:CSch}

Let us now define
\begin{equation}
\label{eq:FgnC}
F_{g,n}^C(t_1,\dots, t_n)
=\sum_{(\mu_1,\dots,\mu_n)\in\bZ_+^n}
D_{g,n}(\mu_1,\dots,\mu_n)\;e^{-\la w,\mu\ra}
\end{equation}
for $2g-2+n>0$.
Then from (\ref{eq:Catalan gn}) we have
$$
W_{g,n}^C(t_1,\dots,t_n) = d_1\cdots d_n
F_{g,n}^C(t_1,\dots,t_n).
$$
Therefore, we have a natural primitive function
of $W_{g,n}(t_1,\dots,t_n)$ for every $(g,n)$.
We note that 
$$
t=-1 \Longrightarrow z=0\Longrightarrow
x=\infty.
$$
Therefore,
\begin{equation}
\label{eq:FgnC zero}
\left.F_{g,n}^C(t_1,\dots,t_n)\right|_{t_i=-1} = 0
\end{equation}
for every $i=1,2,\dots,n$.
The following recursion formula 
of \cite{MZhou}
is the key for
our investigation.

\begin{thm}
\label{thm:FC recursion}
The Laplace transform $F^C_{g,n}(t_{[n]})$ 
satisfies the following
differential recursion equation
for every $(g,n)$ subject to $2g-2+n>0$.
\begin{multline}
\label{eq:FC recursion}
\frac{\partial}{\partial t_1}F^C_{g,n}(t_{[n]})
\\
=
-\frac{1}{16}
\sum_{j=2} ^n
\left[\frac{t_j}{t_1^2-t_j^2}
\left(
\frac{(t_1^2-1)^3}{t_1^2}\frac{\partial}{\partial t_1}
F^C_{g,n-1}(t_{[\hat{j}]})
-
\frac{(t_j^2-1)^3}{t_j^2}\frac{\partial}{\partial t_j}
F^C_{g,n-1}(t_{[\hat{1}]})
\right)
\right]
\\
-\frac{1}{16}
\sum_{j=2} ^n
\frac{(t_1^2-1)^2}{t_1^2}\frac{\partial}{\partial t_1}
F^C_{g,n-1}(t_{[\hat{j}]})
\\
-
\frac{1}{32}\;\frac{(t_1^2-1)^3}{t_1^2}
\left.
\left[
\frac{\partial^2}{\partial u_1\partial u_2}
F^C_{g-1,n+1}(u_1,u_2,t_2, t_3,\dots,t_n)
\right]
\right|_{u_1=u_2=t_1}
\\
-
\frac{1}{32}\;\frac{(t_1^2-1)^3}{t_1^2}
\sum_{\substack{g_1+g_2=g\\
I\sqcup J=\{2,3,\dots,n\}}}
^{\rm{stable}}
\frac{\partial}{\partial t_1}
F^C_{g_1,|I|+1}(t_1,t_I)
\frac{\partial}{\partial t_1}
F^C_{g_2,|J|+1}(t_1,t_J).
\end{multline}
Here we use the index convention 
$[n]=\{1,2,\dots,n\}$ and $[\hat{j}] = 
\{1,2,\dots,\hat{j},\dots, n\}$.
\end{thm}

\begin{rem}
We note that the above formula is 
identical to \cite[Theorem~5.1]{MP2012},
even though $F_{g,n}$ 
 is a different function.
 There we considered the Laplace 
 transform of the number of lattice points
 in $\cM_{g,n}$, and hence $F_{0,1}=
 F_{0,2}=0$.
\end{rem}

\begin{rem}
Because of (\ref{eq:FgnC zero}), 
the recursion
(\ref{eq:FC recursion}) uniquely
determines each $F_{g,n}^C$ by
integrating  from $-1$ to  $t_1$. 
With the same reason, $F_{g,n}^C$ is 
uniquely determined by $W_{g,n}^C$. 
Since we know exactly where $F_{g,n}^C$ 
vanishes, there is no discrepancy of the
constants of integration in (\ref{eq:W=dF}).
\end{rem}

Let us now consider the principal specialization
of the partition function for the Catalan numbers 
\begin{equation}
\label{eq:ZC}
Z^C(t,\hbar)=\exp\left(
\sum_{g=0}^\infty\sum_{n=1}^\infty
\frac{1}{n!}\;
\hbar ^{2g-2+n}F_{g,n}^C(t,t,\dots,t)
\right).
\end{equation}
Since unstable terms $F_{0,1}^C(t)$ and 
$F_{0,2}^C(t,t)$ are included in the above
formula, we need to calculate them first.

\begin{prop}
\label{prop:F01 F02}
In terms of the $z$-variable, we have
\begin{align}
\label{eq:F01C}
F_{0,1}^C(t) &= -\half z^2 +\log z,
\\
\label{eq:F02C}
F_{0,2}^C(t_1,t_2) &= -\log (1-z_1z_2).
\end{align}
\end{prop}
\begin{proof}
Due to the irregularity of $\mu=0$ for 
$D_{0,1}(\mu)$, we need to modify the
definitions (\ref{eq:WgnC}) and  (\ref{eq:FgnC})
for $(g,n) = (0,1)$. It is natural to define
\begin{equation}
\label{eq:W01C}
W_{0,1}(t) = -\sum_{m=0} ^\infty C_{0,1}(2m)
\frac{dx}{x^{2m+1}} = -zdx =\left( -z+\frac{1}{z}
\right)dz
\end{equation}
because of the consistency with (\ref{eq:W01}).
Since $W_{0,1}^C = dF_{0,1}^C$, 
we have 
\begin{align*}
F_{0,1}(t) &= -\half z^2 +\log z +\const
\\
&=
\sum_{m=0} ^\infty D_{0,1}(2m)
\left(\frac{1}{x^{2m}}-\delta_{m,0}
\right),
\end{align*}
where the $m=0$ term is adjusted so that we
do not have the infinity term $D_{0,1}(0)$
in $F_{0,1}^C$.
Using the expression
$$
D_{0,1}(2m) = \frac{1}{2m(m+1)}
\binom{2m}{m},
$$
we have
$$
\lim_{m\rar 0}
D_{0,1}(2m)
\left(\frac{1}{x^{2m}}-\delta_{m,0}
\right) = -\log x.
$$
Since $x=z+\frac{1}{z}$, 
by taking the limit $z\rar 0$ we conclude that
the constant term in $F_{0,1}^C(t)$ is $0$,
which establishes (\ref{eq:F01C}). 

The computation of $F_{0,1}^C(t_1,t_2)$ is 
performed in \cite[Proposition~4.1]{DMSS},
where the idea is to use the 
Euler differential operator
$x_1\frac{d}{dx_1}+x_2\frac{d}{dx_2}$.
By definition $F_{0,1}^C(t_1,t_2)$ does not
have any constant term as $x\rar\infty$, 
therefore there is no constant correction in
(\ref{eq:F02C}), either.
\end{proof}

\begin{thm}
\label{thm:Sch}
The principal specialization of the 
partition function satisfies the following
Schr\"odinger equation.
\begin{equation}
\label{eq:Sch C}
\left(
\hbar^2 \frac{d^2}{dx^2}+\hbar x\frac{d}{dx}+1
\right)
Z^C(t,\hbar) = 0,
\end{equation}
where $t$ is considered as a function in $x$
by
$$
t= t(x) = \frac{z(x)+1}{z(x)-1}
$$
and  {\rm{(\ref{eq:Catalan z})}}.
\end{thm}

The rest of this section is devoted to 
proving the above theorem.
Let us define for $m\ge 0$
\begin{equation}
\label{eq:SmC}
S_m = \sum_{2g-2+n=m-1}\frac{1}{n!}\;
 F^C_{g,n}(t,\dots,t),
\end{equation}
and put
$$
F=\sum_{m=0} ^\infty \hbar^{m-1}\;S_m.
$$
Then the Schr\"odinger equation (\ref{eq:Sch C})
becomes
$$
\hbar^2 \left(
\frac{d^2 F}{dx^2}+\left(\frac{dF}{dx}\right)^2
\right)
+\hbar x\frac{dF}{dx} + 1 = 0,
$$
which is equivalent to 
\begin{equation}
\label{eq:S equation}
\sum_{m=0}^\infty S_m'' \hbar^{m+1}
+\left(\sum_{m=0}^\infty S_m' \hbar^m
\right)^2 +
x\sum_{m=0}^\infty S_m'\hbar ^m+1=0,
\end{equation}
where $'=\frac{d}{dx}$.
Since
$S_0 = F^C_{0,1}$ and 
$$
W^C_{0,1}=dF^C_{0,1} = ydx = -z dx,
$$
we obtain
\begin{equation}
\label{eq:S0'}
S_0' = \frac{d}{dx}F^C_{0,1} = -z
 = -\frac{t+1}{t-1}.
\end{equation}
Using the Lagrangian immersion
(\ref{eq:Catalan x}), we see that
the constant terms of (\ref{eq:S equation})
then become
\begin{equation}
\label{eq:const terms}
(-z)^2 +x(-z) +1=z^2-xz+1 = 0.
\end{equation}

 Collecting 
the $\hbar^{m+1}$-contributions
in (\ref{eq:S equation}) for $m\ge 0$, we obtain
\begin{align*}
-x\frac{d}{dx}S_{m+1}
&=\frac{d^2}{dx^2}S_m +\sum_{a+b=m+1}
\frac{dS_a}{dx}\; \frac{dS_b}{dx}
\\
&=
\frac{d^2}{dx^2}S_m +
2S_0'S_{m+1}'+2S_1'S_m'+
\sum_{\substack{a+b=m+1\\a,b\ge 2}}
\frac{dS_a}{dx}\; \frac{dS_b}{dx}.
\end{align*}
Therefore,
\begin{equation}
\label{eq:S equation 2}
-(2S_0'+x)\frac{d}{dx}S_{m+1}
=
\frac{d^2}{dx^2}S_m +2S_1'\frac{d}{dx}S_m+
\sum_{\substack{a+b=m+1\\a,b\ge 2}}
\frac{dS_a}{dx}\; \frac{dS_b}{dx}.
\end{equation}
To use the closed formulas (\ref{eq:F01C}) and
(\ref{eq:F02C}) for $S_0$ and $S_1$ here, 
we need to  switch
from the $x$-coordinate to the $t$-coordinate.
Using the change of variable formulas
 (\ref{eq:Catalan x}) and (\ref{eq:Catalan t}),
we have
$$
dx = dz\left(1-\frac{1}{z^2}\right)
=-\frac{8t}{(t^2-1)^2}\;dt,
$$
hence 
\begin{equation*}
\begin{aligned}
\frac{d}{dx} &= \frac{z^2}{z^2-1}\;\frac{d}{dz}
=-\frac{(t^2-1)^2}{8t}\;\frac{d}{dt},
\\
\frac{d^2}{dx^2} &= 
\frac{(t^2-1)^4}{64 t^2} \frac{d^2}{dt^2}
+\frac{(t^2-1)^2}{8t}
\left(\frac{d}{dt}\;\frac{(t^2-1)^2}{8t}
\right)\frac{d}{dt}
\\
&=
\frac{(t^2-1)^4}{64 t^2} \frac{d^2}{dt^2}
+\frac{(t^2-1)^3}{64t^3}(3t^2+1)\frac{d}{dt}.
\end{aligned}
\end{equation*}
Therefore, 
\begin{align}
\label{eq:S0'+x}
-(2S_0'+x)\frac{d}{dx}
&=
-\frac{(t^2-1)^2}{8t}\left(z-\frac{1}{z}\right)
\frac{d}{dt}=
-\frac{(t^2-1)^2}{8t}\;\frac{4t}{t^2-1}\;
\frac{d}{dt}
=\half\;(t^2-1)\frac{d}{dt},
\\
\label{eq:S1'}
S_1' &=\half \; \frac{d}{dx}F^C_{0,2}(z,z)
=\half \;
\frac{z^2}{z^2-1}\;\frac{d}{dz}(-\log(1-z^2))
=-\frac{(t^2-1)(t+1)^2}{16 t^2}.
\end{align}
Substituting (\ref{eq:S0'}),
(\ref{eq:S0'+x}) and (\ref{eq:S1'})
in (\ref{eq:S equation 2}), we obtain 
\begin{prop}
For $m\ge 0$, the Schr\"odinger equation
{\rm{(\ref{eq:Sch C})}}
is equivalent to a recursion formula
\begin{equation}
\label{eq:S recursion}
\frac{d}{dt}S_{m+1}
= -\frac{(t^2-1)^3}{32 t^2}
\left[
\frac{d^2}{dt^2}S_m + 
\sum_{\substack{a+b=m+1\\a,b\ge 2}}
\frac{dS_a}{dt}\; \frac{dS_b}{dt}
\right]
-\frac{(t^2-1)^2}{16t^3}(2t^2+t+1)\frac{d}{dt}S_m.
\end{equation} 
\end{prop}

\begin{rem}
The recursion equation 
(\ref{eq:S recursion}) is a direct consequence
of the principal specialization applied to 
(\ref{eq:FC recursion}).
\end{rem}

The proof of (\ref{eq:S recursion}) is given 
in Appendix~\ref{app:Catalan}.

\section{Single Hurwitz numbers}
\label{sect:Hurwitz}

The A-model problem that we are interested 
in this section is the automorphism-weighted
count $H_{g,n}(\mu_1,\dots,\mu_n)$
of the number of the topological types of
meromorphic functions $f:C\lrar \bP^1$
of a nonsingular complete irreducible
 algebraic curve $C$ of genus $g$ that has 
 $n$ labeled poles of orders $(\mu_1,\dots,\mu_n)$
 such that all other critical points of $f$ than 
 these poles are unlabeled simple ramification 
 points. Let $r$ denote the number of such
 simple ramification points. Then 
 the Riemann-Hurwitz formula gives
 \begin{equation}
 \label{eq:RH}
 r=2g-2+n+\sum_{i=1}^n \mu_i.
 \end{equation}
 A remarkable formula due to 
Ekedahl, Lando, Shapiro and Vainshtein
\cite{ELSV, GV1, OP1} 
relates Hurwitz numbers and Gromov-Witten
invariants.  For genus
$g\ge 0$ and the number $n\ge 1$ of the poles
 subject to the 
stability condition
 $2g - 2 +n >0$, the ELSV formula states that
\begin{multline}
\label{eq:ELSV}
H_{g,n}(\mu_1,\dots,\mu_n)=
	 \prod_{i=1}^{n}
	 \frac{\mu_i ^{\mu_i}}{\mu_i!}
	 \int_{\overline{\mathcal{M}}_{g,n}} 
	 \frac{\Lambda_g^{\vee}(1)}
	 {\prod_{i=1}^{n}\big( 1-\mu_i \psi_i\big)}
	 \\
	 =
	 \sum_{j=0}^g (-1)^j
	 \sum_{k_1,\dots,k_n\ge 0}
	 \la \tau_{k_1}\cdots \tau_{k_{n}}c_{j}(\bE)
	 \ra_{g,n} 
	   \prod_{i=1}^{n}
	 \frac{\mu_i ^{\mu_i+k_i}}{\mu_i!},
\end{multline}
where $\overline{\mathcal{M}}_{g,n}$ is 
the Deligne-Mumford moduli stack of stable algebraic
curves of genus $g$ with $n$ distinct smooth
marked points, 
$\Lambda_g^{\vee}(1)=1-c_{1}(\bE)+\cdots +(-1)^g
c_{g}(\bE)$ is the alternating sum of the
Chern classes of the Hodge bundle $\bE$  
 on $\overline{\mathcal{M}}_{g,n}$, 
$\psi_i$ is the $i$-th tautological cotangent class, 
and 
\begin{equation}
\label{eq:LHI}
\la \tau_{k_1}\cdots \tau_{k_n}c_{j}(\bE)\ra_{g,n}
= \int_{\overline{\mathcal{M}}_{g,n}}\psi_1^{k_1}\cdots\psi_n ^{k_n}c_j(\bE)
\end{equation}
is the linear Hodge integral,
which is $0$ unless $k_1+\cdots+k_n +j=3g-3+n$.
Although the Deligne-Mumford stack 
$\Mbar_{g,n}$
is not defined for
$2-2g-n <0$,  single Hurwitz numbers 
are well defined for \emph{unstable} geometries
 $(g,n) = (0,1)$ and $(0,2)$, and their values are
\begin{equation}
\label{eq:unstableHurwitz}
H_0((d)) = \frac{d^{d-3}}{(d-1)!}
=\frac{d^{d-2}}{d!}
\qquad \text{and}\qquad
H_0((\mu_1,\mu_2)) =
 \frac{1}{\mu_1+\mu_2}\cdot
\frac{\mu_1^{\mu_1}}{\mu_1!}\cdot \frac{\mu_2^{\mu_2}}{\mu_2!}.
\end{equation}
The ELSV formula remains valid for unstable cases
if we define
\begin{align}
\label{eq:01Hodge}
&\int_{\overline{\cM}_{0,1}} \frac{\Lambda_0 ^\vee (1)}{1-d\psi}
=\frac{1}{d^2},\\
\label{eq:02Hodge}
&\int_{\overline{\cM}_{0,2}} 
\frac{\Lambda_0 ^\vee (1)}{(1-\mu_1\psi_1)(1-\mu_2\psi_2)}
=\frac{1}{\mu_1+\mu_2}.
\end{align}

The ELSV formula predicts that the single
Hurwitz numbers exhibit the polynomial behavior 
in terms of their Laplace transform. 
Following \cite{EMS}, 
and modifying our choice of the 
coordinates slightly, we define
\begin{multline}
 \label{eq:FgnH}
 F_{g,n}^H(t_1,\dots,t_n)
 =
 \sum_{\mu\in\bZ_+^n} 
 H_{g,n}(\mu_1,\dots,\mu_n)
e^{-\left(
\mu_1w_1+\cdots+\mu_nw_n
\right)}
\\
=
\sum_{\mu\in\bZ_+^n}
\sum_{k_1+\cdots+k_n \le 3g-3+n} \la 
\tau_{k_1}\cdots \tau_{k_n}\Lambda_g^{\vee}(1)
\ra_{g,n}
\left(\prod_{i=1}^{n}
\frac{\mu_i^{\mu_i+k_i}}{\mu_i!}
\right)
e^{-(
\mu_1w_1+\cdots+\mu_n w_n
)}
\\
=
\sum_{k_1+\cdots+k_n \le 3g-3+n} \la 
\tau_{k_1}\cdots \tau_{k_n}\Lambda_g^{\vee}(1)
\ra_{g,n}
\,\,\prod_{i=1}^{n}
\hxi_{k_i}(t_i),
 \end{multline}
 where the $\hxi$-functions are given by
 \begin{equation}
\label{eq:xi}
\hxi_k(t) = 
\sum_{\mu=1}^\infty 
\frac{\mu^{\mu+k}}{\mu!}\;
e^{-\mu w}
=
\sum_{\mu=1}^\infty 
\frac{\mu^{\mu+k}}{\mu!}\;x^\mu
\end{equation}
for $k\ge 0$, and $x=e^{-w}$. Although 
$\hxi_k$ are complicated functions in $x$, their
behavior is simple in terms of $\hxi_0$. So 
we introduce
\begin{equation}
\label{eq:t and z}
t = 1 + 
\sum_{\mu=1}^\infty 
\frac{\mu^{\mu}}{\mu!}\;x^\mu
\qquad {\text{and}}
\qquad 
z =  
\sum_{\mu=1}^\infty 
\frac{\mu^{\mu-1}}{\mu!}\;x^\mu.
\end{equation}
Then by the Lagrange inversion formula, we have
\begin{equation}
\label{eq:Lambert}
x=z e^{-z}\qquad
{\text{and}}
\qquad
z = \frac{t-1}{t},
\end{equation}
and moreover, 
each $\hxi_k(t)$ is a polynomial 
of degree $2k+1$ in $t$ for every $k\ge 0$,
recursively defined by
\begin{equation}
\label{eq:xi recursion}
\hxi_{k+1}(t)=t^2(t-1)\frac{d}{dt}\hxi_k(t).
\end{equation}
This is because 
\begin{equation}
\label{eq:wxt}
-\frac{d}{dw}=x\frac{d}{dx}=t^2(t-1)\frac{d}{dt}.
\end{equation}
Therefore, if $2g-2+n>0$, then 
$ F_{g,n}^H(t_1,\dots,t_n)$ is a symmetric
polynomial of degree $6g-6+3n$, and satisfies
\begin{equation}
\label{eq:FgnH zero}
\left.F_{g,n}^H(t_1,\dots,t_n)\right|_{t_j=1}
=0
\end{equation}
for every $j=1,2\dots,n$.

The computation of \cite{DMSS}
adjusted to our current convention of this paper,
 \begin{equation}
 \label{eq:F01H}
 F_{0,1}^H(t) 
= 
\half\left(1-\frac{1}{t^2}\right)
=z-\half z^2
\end{equation}
and
\begin{equation}
\label{eq:F02H}
F_{0,2}^H(t_1,t_2) 
=\log\left(
\frac{z_1-z_2}
{x_1-x_2}
\right)
-(z_1+z_2),
\end{equation}
determines the Lagrangian immersion
$$
\iota :\Sigma = \bC^* \lrar T^*\bC^* 
$$
by
\begin{equation}
\label{eq:LI for H}
\begin{cases}
x = z e^{-z}\\
y=z,
\end{cases}
\end{equation}
and
$$
W_{0,2}^H(t_1,t_2)=
d_1\tensor d_2 F_{0,2}^H(t_1,t_2)
= \frac{dt_1\cdot dt_2}{(t_1-t_2)^2}
-\frac{dx_1\cdot dx_2}{(x_1-x_2)^2}.
$$
Here the tautological $1$-form on
$T^*\bC^*$ is chosen to be 
\begin{equation}
\label{eq:H eta}
\eta = y\;\frac{dx}{x}.
 \end{equation}
It is consistent with 
$$
W_{0,1}^H(t) = dF_{0,1}^H(t) = z \;\frac{dx}{x}
=\iota^*\eta.
$$
We also note that
\begin{equation}
\label{eq:F01H zero}
\left.F_{0,1}^H(t)\right|_{t=1}=0,
\end{equation}
and 
\begin{equation}
\label{eq:F02H zero}
\left.F_{0,2}^H(t_1,t_2)\right|_{t_j=1}=0,\qquad 
j=1 \text{ or }2.
\end{equation}
The latter equality holds because $t_2=1
\Longrightarrow z_2=0\Longrightarrow
x_2=0$, and hence from
(\ref{eq:F02H}) we have
$$
F_{0,2}^H(t_1,1) = \log\frac{z_1}{x_1}-z_1
=0.
$$

Since
$$
\left.\frac{x_1-x_2}{z_1-z_2}\right|_{z_1=z_2=z}
= (1-z) e^{-z},
$$
the diagonal value of $F_{0,2}^H$ is calculated
as
\begin{equation}
\label{eq:F02H diag}
F_{0,2}^H(t,t) = 
\log\left(\frac{e^{z}}{-z} 
\right)-2z
=-z-\log(1-z) = \frac{1-t}{t}+\log t.
\end{equation}

The single Hurwitz numbers 
$H_{g,n}(\vec{\mu})$ satisfy the
\emph{cut-and-join} equation
\cite{GJ, V}
\begin{multline}
\label{eq:CAJ}
\left(2g-2+n+\sum_{i=1}^n \mu_i
\right)H_{g,n}(\mu_1,\dots,\mu_n)
=
\half\sum_{i\ne j} (\mu_i+\mu_j)
H_{g,n-1}(\mu_i+\mu_j,\mu_{[\hat{i},\hat{j}]})
\\
+
\half
\sum_{i=1}^n \sum_{\a+\b=\mu_i}
\a\b
\left[
H_{g-1,n+1}(\a,\b,\mu_{[\hat{i}]})
+
 \sum_{\substack{g_1+g_2=g\\
I\sqcup J=[\hat{i}]}}
H_{g_1,|I|+1}(\a,\mu_I)H_{g_2,|J|+1}(\b,\mu_J)
\right],
\end{multline}
where we use the convention for indices as in 
Section~\ref{sect:Catalan}.
The Laplace transform of (\ref{eq:CAJ}) is the
polynomial recursion of \cite{MZ} and takes the 
form
\begin{multline}
\label{eq:FgnH recursion}
\left(
2g-2+n+\sum_{i=1}^n t_i(t_i-1)
\frac{\partial}{\partial t_i}
\right)
F_{g,n}^H(t_1,\dots,t_n)
\\
=
\half\sum_{i\ne j}
\frac{t_it_j}{t_i-t_j}
\left(
t_i^2(t_i-1)^2 \frac{\partial}{\partial t_i}
F_{g,n-1}^H(t_{[\hat{j}]})
-
t_j^2(t_j-1)^2 \frac{\partial}{\partial t_j}
F_{g,n-1}^H(t_{[\hat{i}]})
\right)
\\
-\sum_{i\ne j}
t_i^3(t_i-1)\frac{\partial}{\partial t_i}
F_{g,n-1}^H(t_{[\hat{j}]})
\\
+
\half \sum_{i=1}^n 
\left(t_i^2(t_i-1)\right)^2
\left.
\frac{\partial^2}{\partial u_1\partial u_2}
F_{g-1,n+1}^H (u_1,u_2,t_{[\hat{i}]})
\right|_{u_1=u_2=t_i}
\\
+
\half \sum_{i=1}^n 
\left(t_i^2(t_i-1)\right)^2
\sum_{\substack{g_1+g_2=g\\
I\sqcup J=[\hat{i}]}} ^{\text{stable}}
\frac{\partial}{\partial t_i}
F_{g_1,|I|+1}^H(t_i,t_I)\cdot
\frac{\partial}{\partial t_i}
F_{g_1,|J|+1}^H(t_i,t_J).
\end{multline}

It is proved in \cite{EMS} that
the $n$-fold exterior differentiation of 
the above formula is exactly the Eynard-Orantin 
recursion, as predicted by Bouchard and
Mari\~no \cite{BM}. Thus we obtain a natural
integration of the Eynard-Orantin recursion by
taking the Laplace transform of the A-model
quantity again, which is the single Hurwitz umber
$H_{g,n}(\vec{\mu})$.

\begin{thm}
\label{thm:SmH}
Let us define
\begin{equation}
\label{eq:SmH}
S_m^H (t)
= \sum_{2g-2+n=m-1}\frac{1}{n!}F_{g,n}^H
(t,\dots,t).
\end{equation}
Then $S_m^H$'s satisfy
the following second order differential equations:
\begin{equation}
\label{eq:SmH equation}
\left(e^{mw}\frac{d}{dw}e^{-mw}\right)
S_{m+1}^H
=
-\half
\left[
\frac{d^2}{dw^2}S_m^H
+\sum_{a+b=m+1}\frac{dS_a^H}{dw}\cdot 
\frac{dS_b^H}{dw}+\frac{d}{dw}S_m^H
\right].
\end{equation}
Here the $w$-dependence of $t$ is given by
$x=e^{-w}$ and {\rm(\ref{eq:t and z})}.
We also note that $S_m^H(t)$ is a
polynomial of degree 3m-3 for every $m\ge 2$,
and for all values of $m$ we have
\begin{equation}
\label{eq:SmH zero}
\left.S_m^H(t)\right|_{t=1} =0.
\end{equation}
\end{thm}

The proof  is similar to the case of 
the Catalan numbers (Section~\ref{sect:CSch}
and Appendix~\ref{app:Catalan}). 
First we compute
the principal specialization of the differential
equation (\ref{eq:FgnH recursion}).
We then assemble them according to
(\ref{eq:SmH}). By adjusting the unstable 
geometry terms $(g,n) = (0,1)$ and $(0,2)$,
we obtain (\ref{eq:SmH equation}).
However, due to the difference between 
the cut-and-join equation and the
edge-shrinking operation of 
Section~\ref{sect:Catalan}, the resulting
equation becomes quite different.

Choose $m\ge 2$ and $(g,n)$ so that
$2g-2+n = m$. Then the principal specialization
of the left-hand side of (\ref{eq:FgnH recursion})
is
$$
\left(m+t(t-1)\frac{d}{dt}\right)
 F_{g,n}^H(t,\dots,t).
$$
The first line of the right-hand side of 
(\ref{eq:FgnH recursion})
gives
\begin{multline*}
\half t^2 \sum_{i\ne j}
\left.
\frac{\partial}{\partial t_i}
\left(
t_i^2(t_i-1)^2 \frac{\partial}{\partial t_i}
F_{g,n-1}^H(t_i,t,\dots,t)
\right)
\right|_{t_i=t}
\\
=
\frac{n(n-1)}{2}t^2\frac{d}{dt}
\big(t^2(t-1)^2\big)\cdot
\frac{1}{n-1}\frac{d}{dt}
F_{g,n-1}^H(t,\dots,t)
\\
+
\frac{n(n-1)}{2} t^4(t-1) ^2
\left.
\frac{\partial^2}{\partial u^2}
F_{g,n-1}^H(u,t,\dots,t)
\right|_{u=t}
\\
=
\half n!\;t^2\frac{d}{dt}
\big(t^2(t-1)^2\big)\cdot
\frac{1}{(n-1)!}\frac{d}{dt}
F_{g,n-1}^H(t,\dots,t)
\\
+
\half n!(n-1) t^4(t-1) ^2
\frac{1}{(n-1)!}
\left.
\frac{\partial^2}{\partial u^2}
F_{g,n-1}^H(u,t,\dots,t)
\right|_{u=t}.
\end{multline*}
The second line of the right-hand side of 
(\ref{eq:FgnH recursion}) becomes
$$
- n!\; t^3(t-1) \frac{1}{(n-1)!}
\frac{d}{dt}F_{g,n-1}^H(t,\dots,t).
$$
The third line simply produces
$$
\frac{n!}{2}\big(t^2(t-1)\big)^2
\frac{1}{(n+1)!}(n+1)n
\left.
\frac{\partial^2}{\partial u_1\partial u_2}
F_{g-1,n+1}^H(u_1,u_2,t\dots,t)
\right|_{u_1=u_2=t}.
$$
Finally, since the set partition 
becomes the partition of numbers because
all variables are set to be equal,
 the fourth line of the right-hand side of 
(\ref{eq:FgnH recursion}) gives
$$
\frac{n!}{2} \big(t^2(t-1)\big)^2
\sum_{\substack{g_1+g_2=g\\
n_1+n_2=n-1}} ^{\text{stable}}
\frac{1}{(n_1+1)!}
\frac{d}{dt}
F_{g_1,n_1+1}^H(t,\dots,t)\cdot
\frac{1}{(n_2+1)!}
\frac{d}{dt}
F_{g_1,n_2+1}^H(t,\dots,t).
$$
We now apply the operation 
$$
\sum_{2g-2+n=m} \frac{1}{n!}
$$
to the above terms. The left-hand
side becomes 
$$
\left(m+t(t-1)\frac{d}{dt}\right)S_{m+1}^H.
$$
The right-hand side terms are re-assembled
into the sum of $(g',n')$ subject to
$2g'-2+n'=m-1$, following the topological 
structure of the recursion (\ref{eq:FgnH recursion}).
Noticing that unstable geometers are contained
only in $S_0^H$ and $S_1^H$, we obtain
\begin{multline}
\label{eq:SmH recursion}
\left(m+t(t-1)\frac{d}{dt}\right)S_{m+1}^H(t)
\\
=
\half 
\left(t^2 \frac{d}{dt}\big(t^2(t-1)^2\big)
\cdot \frac{d}{dt}S_m^H(t)
+ \big(t^2(t-1)\big)^2
\frac{d^2}{dt^2}S_m^H(t)
- 2t^3(t-1) \frac{d}{dt}S_m^H(t)
\right)
\\
+
\half
\big(t^2(t-1)^2\big)^2
\sum_{\substack{a+b=m+1\\
a,b\ge 2}}
\frac{d}{dt}S_a^H(t)\cdot \frac{d}{dt}S_b^H(t)
\\
=
\half \big(t^2(t-1)\big)^2
\left( 
\frac{d^2}{dt^2}S_m^H(t)
+ \sum_{\substack{a+b=m+1\\
a,b\ge 2}}
\frac{d}{dt}S_a^H(t)\cdot \frac{d}{dt}S_b^H(t)
\right)
+
2t^3(t-1)^2 \frac{d}{dt}S_m^H(t).
\end{multline}

\begin{prop}
The functions $S_m(t)$ are recursively
determined by
\begin{align}
\label{eq:S0H}
S_0^H(t) &= \half\left(1-\frac{1}{t^2}\right),
\\
\label{eq:S1H}
S_1^H(t) &=\half\left(\frac{1-t}{t}+\log t \right),
\end{align}
and
\begin{multline}
\label{eq:SmH integral recursion}
S_{m+1}^H(t) = \left(\frac{t-1}{t}\right)^{-m}
\int_{1} ^t
\Bigg[
\half t^{3-m} (t-1)^{m+1}
\left( 
\frac{d^2}{dt^2}S_m^H(t)
+ \sum_{\substack{a+b=m+1\\
a,b\ge 2}}
\frac{d}{dt}S_a^H(t)\cdot \frac{d}{dt}S_b^H(t)
\right)
\\
+
2t^{2-m}(t-1)^{m+1} \frac{d}{dt}S_m^H(t)
\Bigg]dt.
\end{multline}
\end{prop}

\begin{proof}
As a differential operator,
$$
m+t(t-1)\frac{d}{dt} = 
t(t-1)\left(\frac{t-1}{t}\right)^{-m}
\frac{d}{dt}
\left(\frac{t-1}{t}\right)^{m}.
$$
Therefore, (\ref{eq:SmH recursion}) is equivalent
to
\begin{multline*}
\left[
\left(\frac{t-1}{t}\right)^{-m}
\frac{d}{dt}
\left(\frac{t-1}{t}\right)^{m}
\right]S_{m+1}^H(t)
\\
=
\half t^3(t-1)
\left( 
\frac{d^2}{dt^2}S_m^H(t)
+ \sum_{\substack{a+b=m+1\\
a,b\ge 2}}
\frac{d}{dt}S_a^H(t)\cdot \frac{d}{dt}S_b^H(t)
\right)
+
2t^2(t-1) \frac{d}{dt}S_m^H(t).
\end{multline*}
On the right-hand side only $S_k^H(t)$ with 
$k\le m$ appear. Using the
fact of the zero  (\ref{eq:SmH zero}),
we obtain (\ref{eq:SmH integral recursion}).
\end{proof}

On the third line of (\ref{eq:SmH recursion})
the terms with $S_0^H$ and $S_1^H$ are
not included. More precisely, these omitted
terms are
$$
t(t-1)^2 \frac{d}{dt}S_{m+1}^H(t)
+\half t^2(t-1)^3  \frac{d}{dt}S_{m}^H(t).
$$
When we add these terms
 to (\ref{eq:SmH recursion}),
and adjust the second order differentiation as
$$
\left( t^2(t-1)\frac{d}{dt}\right)^2
= \big(t^2(t-1)\big)^2 \frac{d^2}{dt^2}
+ t^3(t-1)(3t-2)\frac{d}{dt},
$$
we finally obtain
\begin{multline}
\label{eq:SmH equation in t}
\left(m+t^2(t-1)\frac{d}{dt}\right)S_{m+1}^H(t)
\\
=
\half 
\left( 
\left(t^2(t-1)\frac{d}{dt}\right)^2 S_m^H(t)
+ \big(t^2(t-1)\big)^2\sum_{a+b=m+1}
\frac{dS_a^H(t)}{dt}\cdot \frac{dS_b^H(t)}{dt}
-
t^2(t-1) \frac{d}{dt}S_m^H(t)
\right).
\end{multline}
Then (\ref{eq:SmH equation}) follows from
(\ref{eq:SmH equation in t}) and (\ref{eq:wxt}).
This completes the proof 
Theorem~\ref{thm:SmH}.

\begin{thm}
\label{thm:Sch H}
Let us define
the  partition function for the single Hurwitz 
numbers  in a similar way:
\begin{equation}
\label{eq:ZH}
Z^H(t,\hbar)=
\exp\left(
\sum_{m=0} ^\infty S_m^H \hbar^{m-1}
\right)
=
\exp\left(
\sum_{g=0} ^\infty \sum_{n=1}^\infty
\frac{1}{n!}\; \hbar^{2g-2+n}F_{g,n}^H(t,\dots,t)
\right).
\end{equation}
Then
 we have
\begin{equation}
\label{eq:Sch H}
\left[
\frac{1}{2}\; \frac{\partial^2}{\partial w^2}
+\left(\frac{1}{2}+\frac{1}{\hbar}\right)
\frac{\partial}{\partial w} 
-\frac{\partial}{\partial \hbar}
\right] Z^H(t,\hbar)=0.
\end{equation}
\end{thm}

\begin{rem}
Eq.(\ref{eq:Sch H}) is a heat equation,
where $\hbar$ is considered as the
\emph{time} variable of the heat conduction. 
It determines the solution uniquely with the
``initial condition''
$$
\left.Z^H\big(t(w),\hbar\big)\right|_{\hbar\sim 0}
=\exp\left(\frac{1}{\hbar}S_0^H+S_1^H\right)
$$
given by (\ref{eq:S0H}) and (\ref{eq:S1H}).
\end{rem}

\begin{proof}
Let 
$$
F^H(t,\hbar)=
\sum_{m=0} ^\infty S_m^H \hbar^{m-1}
=
\sum_{g=0} ^\infty \sum_{n=1}^\infty
\frac{1}{n!}\; \hbar^{2g-2+n}F_{g,n}^H(t,\dots,t).
$$
In terms of $F^H$, (\ref{eq:Sch H}) is equivalent
to 
\begin{equation}
\label{eq:Sch FH}
\frac{\hbar}{2}
\frac{\partial^2 F^H}{\partial w^2}
+
\frac{\hbar}{2}
\left(\frac{\partial F^H}{\partial w}\right)^2
+
\left[\left(1+\frac{\hbar}{2}\right)
\frac{\partial}{\partial w} 
-\hbar\frac{\partial}{\partial \hbar}
\right] F^H=0.
\end{equation}
We apply the operation 
$$
\sum_{m=0}^\infty \hbar^m
$$
to (\ref{eq:SmH equation}).
The left-hand side is
\begin{equation}
\label{eq:SmH lhs}
\sum_{m=0}^\infty \hbar^m \left(
-m + \frac{d}{dw}\right)S_{m+1}^H
=\left(
\frac{\partial}{\partial w}-
\hbar\frac{\partial}{\partial \hbar}
\right)
\left(
F^H
-\frac{1}{\hbar}S_0^H\right).
\end{equation}
The right-hand side gives
\begin{multline}
\label{eq:SmH rhs}
-\frac{\hbar}{2}
\left(
\frac{\partial^2}{\partial w^2}
\sum_{m=0} ^\infty S_m^H \hbar^{m-1}
+
\sum_{m=0} ^\infty
\sum_{a+b=m+1}
\frac{\partial S_a^H}{\partial w}\hbar^{a-1}
\cdot
\frac{\partial S_b^H}{\partial w}\hbar^{b-1}
+
\frac{\partial}{\partial w}
\sum_{m=0} ^\infty S_m^H \hbar^{m-1}
 \right)
 \\
 =
 -\frac{\hbar}{2}
 \left(
 \frac{\partial^2}{\partial w^2} F^H
 + \left(\frac{\partial F^H}{\partial w}\right)^2
 -\frac{1}{\hbar^2}
 \left(\frac{\partial S_0^H}{\partial w}\right)^2
 + \frac{\partial}{\partial w} F^H
 \right).
\end{multline}
If we collect all terms that contain $S_0^H(t)$ in 
(\ref{eq:SmH lhs}) and (\ref{eq:SmH rhs}),
we have an equation
$$
\left(t^2(t-1)\frac{\partial}{\partial t}
+\hbar\frac{\partial}{\partial \hbar}
\right) \frac{1}{\hbar}S_0^H
=
\frac{1}{2\hbar} 
\left(t^2(t-1)\frac{\partial }{\partial t}S_0^H
\right)^2,
$$
or equivalently, 
\begin{equation}
\label{eq:S0H equation}
S_0^H(t) = t^2(t-1)\frac{d}{dt}S_0^H(t)
-\half \big(t^2(t-1)\big)^2 
\left(\frac{dS_0^H(t)}{dt}\right)^2.
\end{equation}
We see that (\ref{eq:S0H}) is a solution to
(\ref{eq:S0H equation}).
After eliminating (\ref{eq:S0H equation})
from (\ref{eq:SmH lhs}) $=$ (\ref{eq:SmH rhs}),
we obtain (\ref{eq:Sch FH}). This completes
the proof.
\end{proof}

\section{The Schur function expansion of
the Hurwitz partition function}

\label{sect:Schur}

Let us  introduce the \emph{free energy}
of single Hurwitz numbers 
by a formal sum as
\begin{multline}
\label{eq:Hfree energy}
F^H(t_1,t_2,t_3,\dots;\hbar)
=
\sum_{g\ge 0, \; n\ge 1}
\frac{1}{n!} \;\hbar^{2g-2+n} \;
F_{g,n}^H(t_1,\dots,t_n)
\\
=
\sum_{g\ge 0, \; n\ge 1}
\frac{1}{n!} \;\hbar^{2g-2+n} 
\sum_{(\mu_1\dots,\mu_n)\in\bZ_+^n}
H_{g,n}(\mu_1\dots,\mu_n)\;
e^{-(\mu_1+\cdots+\mu_n)} 
\prod_{i=1}^n 
x_i ^{\mu_i}.
\end{multline}
The partition function we considered in
Section~\ref{sect:Hurwitz}
is the principal specialization
$$
Z^H(t,\hbar) = \exp\left(F^H(t,t,\dots:\hbar)\right).
$$

Recall 
 another  generating function of the 
Hurwitz numbers  
\cite{Kazarian,KazarianLando,O}
defined by
\begin{equation}
\label{eq:H}
\mathbf{H}(s,\bp) = \sum_{g\ge 0,\;n\ge 1}
\mathbf{H}_{g,n} (s,\mathbf{p}), 
\end{equation}
\begin{equation}
\label{eq:Hgn}
\mathbf{H}_{g,n} (s,\mathbf{p}) = 
\frac{1}{n!}
\sum_{\vec{\mu}\in\bZ_+^n} H_{g,n}(\vec{\mu}) 
\mathbf{p}_\mu s^{r(g,\mu)},
\end{equation}
where $\bp_\mu = p_{\mu_1}\cdots p_{\mu_n}$,
and 
$$
r = r(g,\mu)=2g-2+n + \sum_{i=1}^n \mu_i
$$
is again 
the number of simple ramification point of a 
Hurwitz cover of genus $g$ and profile $\mu$.

At this point we wish to go back and forth between 
the following 
two distinct points of view: One is to regard
$\vec{\mu}
=(\mu_1,\dots,\mu_n)$ as a vector consisting of
positive integers, and the other is to view $\mu$
as a \emph{partition} of length $n$.
For any function $f(\vec{\mu})$ in $\vec{\mu}$
as a vector, we
have a change of summation formula
\begin{equation}
\label{eq:resum}
\sum_{\vec{\mu}\in\bZ_+^n}f(\vec{\mu})
=
\sum_{\mu:\ell(\mu) = n}\frac{1}{\big|\Aut(\mu)\big|}
\sum_{\sigma\in S_n}f(\vec{\mu}_\sigma).
\end{equation}
Here the first sum in the right-hand side 
runs over partitions $\mu$ of 
a fixed length $n$,
the second sum is over
the symmetric group $S_n$ of $n$ letters,
$$
\vec{\mu}_\sigma = \left(\mu_{\sigma(1)},
\dots,\mu_{\sigma(n)}
\right)\in \bZ_+^n
$$ 
is the integer vector obtained by
permuting the parts of $\mu$ by $\sigma\in S_n$,
and $\Aut(\mu)$ is the permutation group
interchanging the equal parts of $\mu$.
As a partition, the length of $\mu$ is denoted
by $\ell(\mu)$, and its \emph{size}
by
$$
|\mu|=\sum_{i=1}^{\ell(\mu)} \mu_i.
$$
Often single Hurwitz numbers are labeled
by the genus $g$ and a \emph{partition}
$\mu$. In this case the expression
$$
h_{g,\mu} = \frac{r(g,\mu)!}{|\Aut(\mu)|}
H_{g,\ell(\mu)}(\mu_1,\dots,\mu_{\ell(\mu)})
$$
is used in the literature,
when we do not label the poles, but  
label the simple ramification points.
The generating function 
then has an expression in terms of
sumes over partitions:
$$
\mathbf{H}(s,\bp)
=\sum_{g=0}^\infty \sum_\mu 
h_{g,\mu}\bp_\mu \frac{s^{r(g,\mu)}}{r(g,\mu)!}.
$$

In terms of the  ELSV formula (\ref{eq:ELSV}) we have
\begin{multline}
\label{eq:HinELSV}
\mathbf{H}(s,\bp) \\
=\sum_{g\ge 0,\;n\ge 1}  \frac{1}{n!}\;
s^{2g-2+n}
\sum_{ k_1+\cdots +k_n \leq 3g-3+n} 
\la \tau_{k_1}\cdots \tau_{k_n}\Lambda_g^{\vee}
(1)\ra_{g,n}\;
\prod_{i=1}^{n}\sum_{\mu_i=1}^{\infty}
\frac{\mu_i^{\mu_i+k_i}}{\mu_i!}s^{\mu_i} p_{\mu_i}
\end{multline}
with an appropriate incorporation of
 (\ref{eq:01Hodge})
and (\ref{eq:02Hodge}).
Recall the Laplace transform (\ref{eq:FgnH})
here for comparison that is assembled into
the \emph{free energy}
\begin{multline}
\label{eq:FH in ELSV}
F^H(t_1,t_2\dots;\hbar)
:=
\sum_{g\ge 0, \; n\ge 1}
\frac{1}{n!} \;\hbar^{2g-2+n} \;
F_{g,n}^H(t_1,\dots,t_n)
\\
=
\sum_{g\ge 0,n\ge 1}
\frac{1}{n!}\;\hbar^{2g-2+n}
\sum_{ k_1+\cdots +k_n \leq 3g-3+n} 
\la \tau_{k_1}\cdots \tau_{k_n}\Lambda_g^{\vee}
(1)\ra_{g,n}
\prod_{i=1}^{n}\sum_{\mu_i=1}^{\infty}
\frac{\mu_i^{\mu_i+k_i}}{\mu_i!} \;
x_i^{\mu_i}.
\end{multline}
It is easy to see that the relation between 
the two sets of variables is exactly the 
power-sum symmetric functions.
Let us re-scale the  usual power-sum symmetric
function  $p_j$   of degree $j$ in $x_i$'s
with a scale parameter $s$ as follows: 
\begin{equation}
\label{eq:p(s)}
p_j(s) := s^{-j}\left(x_1^j+x_2^j+x_3^j +\cdots\right).
\end{equation}
Here we consider $p_j(s)$ as a degree $j$ polynomial
defined on $\bC^n$, but the dimension $n$ is unspecified.
Then for every $\mu\in\bZ_+^n$, we have
\begin{equation}
\label{eq:pderivative}
d_1\cdots d_n \;\bp_\mu(s)
=
s^{-(\mu_1+\cdots+\mu_n)}
\left(
\sum_{\sigma\in S_n} \prod _{i=1}^n
\mu_i \;x_{\sigma(i)} ^{\mu_i-1} 
\right)
dx_1\cdots dx_n
\end{equation}
as a differential form on $\bC^n$.

Now from (\ref{eq:p(s)}), 
(\ref{eq:pderivative}) and (\ref{eq:resum}),
we obtain
\begin{multline}
\label{eq:W=dH}
d_1\cdots d_n\; \mathbf{H}_{g,n}(s,\bp(s))
=
 \sum_{\mu:\ell(\mu)=n}
h_{g,\mu}\;d_1\cdots d_n\;\bp_\mu(s) \;
\frac{s^r}{r!}
\\
=
 \sum_{\mu:\ell(\mu)=n}
h_{g,\mu}\;
\frac{s^r}{r!} \; s^{-|\mu|}
\sum_{\sigma\in S_n} \prod _{i=1}^n
\mu_i \;x_{\sigma(i)} ^{\mu_i-1} dx_1\cdots dx_n
\\
=
 \sum_{\mu:\ell(\mu)=n}
\frac{1}{|\Aut(\mu)|}\;H_{g,n}(\mu_1,\dots,\mu_n)
s^{2g-2+n}
\sum_{\sigma\in S_n} \prod _{i=1}^n
\mu_i \;x_{\sigma(i)} ^{\mu_i-1} dx_1\cdots dx_n
\\
=
 \sum_{\mu\in\bZ_+^n}
H_{g,n}(\mu_1,\dots,\mu_n)
s^{2g-2+n}
 \prod _{i=1}^n
\mu_i \;x_{i} ^{\mu_i-1} dx_1\cdots dx_n
\\
= s^{2g-2+n} d_1\cdots d_n
\sum_{\mu\in\bZ_+^n}
H_{g,n}(\mu_1,\dots,\mu_n)
 \prod _{i=1}^n
x_{i} ^{\mu_i}
\\
=
 s^{2g-2+n} d_1\cdots d_n
F_{g,n}^H(t_1,\dots,t_n)
=
 s^{2g-2+n}\;
 W_{g,n}^H(t_1,\dots,t_n).
\end{multline}
This formula tells us that 
the Eynard-Orantin differential form
$W_{g,n}^H$ is the exterior
derivative of $\mathbf{H}_{g,n}(s,\bp(s))$
with the  identification (\ref{eq:p(s)}).
Moreover, 
we have
\begin{equation}
\label{eq:FgnH=Hgn}
s^{2g-2+n}\;
F_{g,n}^H(t_1,\dots,t_n) \equiv
\mathbf{H}_{g,n}(s,\bp(s))  \mod
\Ker(d_1\cdots d_n)
\end{equation}
as functions on $\bC^n$.

\begin{rem}
Let us examine (\ref{eq:FgnH=Hgn}).
For $n=1$, the power sum
(\ref{eq:p(s)}) contains only one term
 and we have
\begin{equation}
\label{eq:F1=H1}
\mathbf{H}_{g,1}(s,\bp(s))
=
s^{2g-1}
\sum_{k=1} ^\infty H_{g,1}(k) p_k(s) s^k
=
s^{2g-1}
\sum_{k=1} ^\infty H_{g,1}(k) \; x^k
=
s^{2g-1} F^H_{g,1}(t_1).
\end{equation}
Thus (\ref{eq:FgnH=Hgn}) is an equality for $n=1$.
In general what happens is
\begin{multline}
\label{eq:FgnH=Hgnex}
\mathbf{H}_{g,n}(s,\bp)
=
\frac{1}{n!}\;
s^{2g-2+n}
\sum_{\vec{\mu}\in\bZ_+^n}
H_{g,n}(\vec{\mu})  \prod_{i=1}^n \left(
x_1^{\mu_1}+x_2^{\mu_2}+\cdots
+x_n^{\mu_n}\right)
\\
=
s^{2g-2+n}
\sum_{\vec{\mu}\in\bZ_+^n}
H_{g,n}(\vec{\mu})\;
x_1^{\mu_1}x_2^{\mu_2}\cdots x_n^{\mu_n}
+(\text{terms with less than $n$ variables}).
\end{multline}
Therefore, (\ref{eq:FgnH=Hgn}) is never an
equality for $n>1$.
\end{rem}

However, the principal specialization  
$t=t_1=t_2=t_3=\cdots$ corresponds to 
evaluating
\begin{equation}
\label{eq:pj=xj}
p_j =
\left( \frac{x}{s}\right)^j.
\end{equation}
With this identification we have again an equality
\begin{multline}
\label{eq:H=FH}
\mathbf{H}(s,\bp)
\\
= \sum_{g\ge 0,n\ge 1} \frac{1}{n!}\;s^{2g-2+n}
\sum_{\vec{\mu}\in\bZ_+^n}
H_{g,n}(\vec{\mu}) \; x^{(\mu_1+\cdots+\mu_n)}
=
\sum_{g\ge 0,n\ge 1} \frac{1}{n!}\;s^{2g-2+n}
F_{g,n}^H(t,t,\dots,t).
\end{multline}

In Section~\ref{sect:Hurwitz} we noted
that the Eynard-Orantin recursion for 
Hurwitz numbers is the Laplace transform of
the cut-and-join equation (\ref{eq:CAJ}) 
\cite{EMS,MZ}. Another 
consequence of the same combinatorial equation is
a \emph{heat equation}
 \cite{Goulden,Kazarian, Zhou1}
\begin{equation}
\label{eq:heat CAJ}
\frac{\partial}{\partial s}
e^{\mathbf{H}(s,\mathbf{p})}
=
\frac{1}{2}
\left[
\sum_{i,j\ge 1} \left(
(i+j)p_ip_j\frac{\partial}{\partial p_{i+j}}
+ijp_{i+j}\frac{\partial^2}{\partial p_i \partial p_j}
\right)
\right]
 e^{\mathbf{H}(s,\mathbf{p})} ,
\end{equation}
with the initial condition $\mathbf{H}(0,\bp) = p_1$. 
An important and fundamental fact here is that 
the heat equation (\ref{eq:heat CAJ})  implies
that 
$e^{\mathbf{H}(s,\bp)}$ is  a KP $\tau$-function 
for each value of $s$
\cite{Kazarian, KazarianLando, O, Zhou1}.
Let us recall this fact here. 
 We note that a solution of the heat equation 
 is expanded by the eigenfunction of the
 second order operator. In our case of
 (\ref{eq:heat CAJ}), the eigenfunctions 
 of the cut-and-join operator on the 
 right-hand side  are given by the
 \emph{Schur functions}.

 For a partition $\mu = (\mu_1\ge \mu_2\ge\cdots)$ 
 of a finite length $\ell(\mu)$,
we define the \textbf{shifted power-sum function}
by
\begin{equation}
\label{eq:shifted power-sum}
\bp_r[\mu] := \sum_{i=1}^\infty
\left[
\left(\mu_i-i+\half\right)^r-
\left(-i+\half\right)^r
\right].
\end{equation}
This is a finite sum of $\ell(\mu)$ terms.
In this paper we consider $\bp_r[\mu]$ as
a number associated with a partition $\mu$. 
Then we have \cite{Goulden, Zhou1}
\begin{equation}
\label{eq:Schur}
\sum_{i,j\ge 1} \left(
(i+j)p_ip_j\frac{\partial}{\partial p_{i+j}}
+ijp_{i+j}\frac{\partial^2}{\partial p_i \partial p_j}
\right)
s_\mu(\bp)
= \bp_2[\mu] 
\cdot 
s_\mu(\bp),
\end{equation}
where $s_\mu(\bp)$ is the Schur function 
defined by
\begin{align*}
s_\mu(\bp) &= \sum_{|\lam|=|\mu|}
\frac{\rchi_\mu(\lam)}{z_\lam} \bp_\lam,
\\
z_\mu &= \prod_{i=1}^{\ell(\mu)} m_i! i^{m_i},
\\
m_i &={\text{ the number of parts in $\mu$ of 
length $i$,}}
\end{align*}
and $\rchi_\mu(\lam)$ is  the value of the irreducible 
character of the representation $\mu$ 
of the symmetric group evaluated at 
the conjugacy class $\lam$.
If we write
$$
\triangle =\half \sum_{i,j\ge 1} \left(
(i+j)p_ip_j\frac{\partial}{\partial p_{i+j}}
+ijp_{i+j}\frac{\partial^2}{\partial p_i \partial p_j}
\right),
$$
Then
$$
\triangle s_\mu(\bp) = \half \bp_2[\mu]
\cdot s_\mu(\bp)
$$
and 
$$
\frac{\partial}{\partial s} e^{\mathbf{H}(s,\bp)}
=\triangle e^{\mathbf{H}(s,\bp)}.
$$
Therefore, we have an expansion formula
$$
e^{\mathbf{H}(s,\bp)}
=\sum_\mu c_\mu s_\mu(\bp)e^{\half \bp_2[\mu] s}
$$
for a constant $c_\mu$ associated with every 
partition $\mu$. The constants are determined by
the initial value. Since the initial condition is
$$
\left.e^{\mathbf{H}(s,\bp)}\right|_{s=0}
=e^{p_1} = \sum_\mu c_\mu s_\mu(\bp),
$$
we conclude that
$$
c_\mu = s_\mu(1,0,0,\dots,0).
$$
This follows from the Cauchy identity
\begin{equation}
\label{eq:Cauchy}
\frac{1}{\prod_{i,j} (1-x_iy_j)}
=\sum_{\mu}s_\mu(\bp)s_\mu(\bp^y),
\end{equation}
where 
$$
p_j = \sum_{i}x_i ^j
\qquad{\text{and}}\qquad
p_j^y = \sum_{i}y_i ^j.
$$
Since we have
\begin{multline}
\label{eq:Cauchy 1}
\sum_{\mu}s_\mu(\bp)s_\mu(\bp^y)
=
\frac{1}{\prod_{i,j} (1-x_iy_j)}
=\exp\left( -\sum_{i,j}\log(1-x_iy_j)\right)
\\
=\exp\left(\sum_{i,j}\sum_{m\ge 1}\frac{1}{m}
x_j^m y_i^m\right)
=\exp\left( \sum_{m\ge 1} \frac{1}{m}
p_m p_m^y\right),
\end{multline}
the restriction of (\ref{eq:Cauchy 1})
to $p_1^y =1$ and $p_m^y=0$ for
all $m\ge 2$ reduces to
$$
e^{p_1} = \sum_\mu s_\mu(1,0,\dots,0)s_\mu(\bp).
$$
Because of the determinantal formula for
the Schur functions, $s_\mu(1,0,0,\dots,0)$'s
are the Pl\"ucker coordinate of a point of the
Sato Grassmannian. It follows that
$$
s_\mu(1,0,0,\dots,0)
e^{\half \bp_2[\mu] s}
$$ 
for all $\mu$  also form the
Pl\"ucker coordinates because
of (\ref{eq:shifted power-sum}).
Then by a theorem of Sato \cite{Sato},
\begin{equation}
\label{eq:eH as tau}
e^{\mathbf{H}(s,\bp)}
=\sum_\mu s_\mu(1,0,0,\dots,0)
 s_\mu(\bp)e^{\half \bp_2[\mu] s}
 =\sum_\mu \frac{\dim\mu}{|\mu|!}
 e^{\half \bp_2[\mu] s}s_\mu(\bp)
\end{equation}
is a $\tau$-function of the KP equations.
Here $\dim\mu$ is the dimension of the 
irreducible representation of the symmetric
group $S_{|\mu|}$ belonging to the partition 
$\mu$.

Thus we have established

\begin{thm}
\label{thm:ZH=e^H}
The Hurwitz partition function $Z^H(t,\hbar)$
of {\rm{(\ref{eq:ZH})}}
 is obtained by evaluation of
 the KP $\tau$-function 
 $e^{\mathbf{H}(s,\bp)}$
 at  $(s,\bp)=(\hbar,\bp(\hbar))$:
 \begin{equation}
\label{eq:ZH=eH}
Z^H(t,\hbar) =
e^{\mathbf{H}(\hbar,\bp(\hbar))}.
\end{equation}
Here $\bp(\hbar)$
 means the principal specialization
\begin{equation}
\label{eq:pj=xjh}
p_j =
\left( \frac{x}{\hbar}\right)^j
\end{equation}
for every $j=1,2,3,\dots.$
The $t$-variable and the $x$-variable are related by
$$
x = \frac{t-1}{t}e^{\frac{1}{t}-1}.
$$
\end{thm}

Since we have a concrete 
expansion formula (\ref{eq:eH as tau})
for $e^{\mathbf{H}(s,\bp)}$, it is straightforward 
to find a formula for its principal specialization. 
Let us  look at (\ref{eq:Cauchy 1}) again. 
This time we apply the principal specialization 
to both $\bp$ and $\bp^y$, meaning that
we substitute 
$$
p_m=x^m \qquad {\text{ and}}\qquad
 p_m^y = y^m.
 $$
Then we have
$$
\sum_{\mu}s_\mu(\bp)s_\mu(\bp^y)
=\sum_{m=0}^\infty x^my^m
$$
after the double principal specialization. 
Therefore, the sum with respect to all partitions
$\mu$ is reduced to the sum with respect to
only one-part partitions,
i.e., $\mu = (m)$. All other partitions contribute $0$.
Thus $s_\mu=s_m=h_m$, which is the $m$-th
complete symmetric function. But because of
the principal specialization, we simply
have $h_m = x^m$. 
We note that if $\mu = (m)$, then 
(\ref{eq:shifted power-sum}) reduces to 
$$
\bp_2[(m)] = \left(m-1+\half\right)^2-
\left(-1+\half\right)^2 = m(m-1).
$$
We therefore conclude that
\begin{equation}
\label{eq:ZH in x}
Z^H(t,\hbar) = 
e^{\mathbf{H}(\hbar,\bp(\hbar))}
=\sum_{m=0}^\infty \frac{1}{m!}
e^{\half m(m-1)\hbar}\left(\frac{x}{\hbar}\right)^m.
\end{equation}
We have now established a theorem of Zhou
\cite{Zhou4}.

\begin{thm}[\cite{Zhou4}]
\label{thm:Sch H2}
The same Hurwitz partition function satisfies
a differential-difference equation
\begin{equation}
\label{eq:Sch H2}
\left(
\hbar \frac{\partial}{\partial w}+
e^{-w} e^{-\hbar\frac{\partial}{\partial w}}
\right)Z^H(t,\hbar) = 0.
\end{equation}
Here again $t$ is a function in $w$, 
which is given by
$$
t = 1 +\sum_{k=1}^\infty \frac{k^k}{k!}x^k
=1 +\sum_{k=1}^\infty \frac{k^k}{k!}e^{-wk}.
$$
The characteristic variety of this equation, 
with the identification of 
$$
z = -\hbar\frac{\partial}{\partial w},
$$
is the Lagrangian immersion $e^{-w} = ze^{-z}$.
\end{thm}

\begin{proof} 
Let us denote
$$
a_m=
e^{\half m(m-1)\hbar}\left(\frac{x}{\hbar}\right)^m
=e^{(1+2+\cdots+(m-1))\hbar}
\left(\frac{x}{\hbar}\right)^m.
$$
Then 
$$
Z^H(t,\hbar) = \sum_{m=0}^\infty \frac{1}{m!}
a_m, 
\qquad
a_{m+1} = 
e^{m\hbar}a_m\frac{x}{\hbar},
\qquad{\text{and}}
\qquad
x\frac{d}{dx}a_m = ma_m.
$$
We note that $x\frac{d}{dx}$ operates as
the multiplication of $m$ to $a_m$.
Therefore,
\begin{multline*}
-\hbar\frac{\partial}{\partial w}Z^H(t,\hbar)
=
\hbar x\frac{d}{dx}\sum_{m=0}^\infty 
\frac{1}{m!}
a_m
=\hbar\sum_{m=0}^\infty \frac{1}{m!}a_{m+1}
=x 
\sum_{m=0}^\infty \frac{1}{m!}e^{m\hbar}a_m
\\
=x e^{\hbar x\frac{d}{dx}} Z^H(t,\hbar)
=e^{-w}e^{-\hbar\frac{\partial}{\partial w}}
Z^H(t,\hbar).
\end{multline*}
This completes the proof.
\end{proof}

\begin{rem}
The asymptotic behavior of $Z^H(t,\hbar)$ 
near $\hbar=0$ is determined by 
$$
\exp\left(\frac{1}{\hbar} S_0^H+S_1^H\right).
$$
From (\ref{eq:Sch H2}) we have
\begin{multline*}
0=\left. 
\exp\left(-\frac{1}{\hbar} S_0^H-S_1^H\right)
\left(
\hbar \frac{d}{dw}
+e^{-w}e^{-\hbar \frac{d}{dw}}
\right) 
\exp\left(\frac{1}{\hbar} S_0^H+S_1^H\right)
\right|_{\hbar = 0}
\\
=
\frac{d}{dw}S_0^H\big(t(w)\big)+
\left.
e^{-w}\exp\left(\frac{S_0^H\big(t(w-\hbar)\big)
-S_0^H\big(t(w)\big)}{\hbar}\right)
\right|_{\hbar=0}
\\
=
\frac{d}{dw}S_0^H
+e^{-w} e^{-\frac{d}{dw}S_0^H}
=-z + xe^z.
\end{multline*}
We thus recover the Lagrangian immersion 
in this way as well.
\end{rem}

\begin{rem}
We can directly verify that the
 expression (\ref{eq:ZH in x}) 
 of $Z^H(t,\hbar)$ satisfies the
 Schr\"odinger equation (\ref{eq:Sch H}). 
 Indeed, since $x=e^{-w}$, 
 \begin{multline*}
 \left[
 \half \frac{\partial^2 }{\partial w^2}+
 \left(\half +\frac{1}{\hbar}\right) 
 \frac{\partial}{\partial w}-\frac{\partial}{\partial 
 \hbar}
 \right]
 \sum_{m=0}^\infty \frac{1}{m!}e^{\half m(m-1)
 \hbar} \frac{1}{\hbar^m}e^{-mw}
 \\
 =
 \sum_{m=0}^\infty \frac{1}{m!}e^{\half m(m-1)
 \hbar}\frac{1}{\hbar^m}e^{-mw}
 \left[
 \half m^2 -m\left(\half +\frac{1}{\hbar}\right)
 -\half m(m-1) +\frac{m}{\hbar}
 \right] = 0.
 \end{multline*}
\end{rem}

\section{Conclusion}
\label{sect:conclusion}

The main purpose of this paper is to 
derive the Schr\"odinger equation of the
partition function from the integrated 
Eynard-Orantin topological recursion,
when there is an A-model counting problem
whose mirror dual is the Eynard-Orantin theory. 
We examined two different types of counting
problem of ramified covering of $\bP^1$: 
one is Grothendieck's dessins d'enfants, 
and the other single Hurwitz numbers. 
The first example leads to a Lagrangian 
immersion in $T^*\bC$ defined by
a Laurent polynomial
equation (\ref{eq:Catalan immersion}),
 while the latter
corresponds to the Lambert curve in $T^*\bC^*$
given by an exponential equation
(\ref{eq:LI for H}).

If we start with the Eynard-Orantin recursion
(\ref{eq:EO}) and define the primitive functions
$F_{g,n}$ by (\ref{eq:W=dF}), then we have the
ambiguity in the constants of integration. 
However, if we start with an A-model,
 then the primitive functions 
$F_{g,n}$'s are given by the Laplace transform 
of the solution to the A-model problem. 
The examples we have studied in this paper 
show that always there is a natural zero
$t_i=a$
of $F_{g,n}(t_1,\dots,t_n)$ in each 
variable, such as (\ref{eq:FgnC zero})
and (\ref{eq:FgnH zero}), when $2g-2+n>0$. 
Thus the integration
formula
$$
F_{g,n}(t_1,\dots,t_n)
=\int_a ^{t_1}\cdots\int_a^{t_n}
W_{g,n}(t_1,\dots,t_n)
$$
uniquely determines $F_{g,n}$ from $W_{g,n}$. 
The integral transform equation (\ref{eq:EO}) 
for $W_{g,n}$ is then equivalent to a
differential transform equation for $F_{g,n}$,
such as (\ref{eq:FC recursion}) and 
(\ref{eq:FgnH recursion}) for our examples.
Because of our assumption for the 
Lagrangian immersion that the 
Lagrangian singularities are simply ramified,
the differential recursion equation for $F_{g,n}$
is expected to be a second order PDE. 
Then by taking the principal specialization, 
we obtain a second order differential equation
in $t$ and $\hbar$. 
For our examples we have thus established
(\ref{eq:Sch C}) and (\ref{eq:Sch H}).

Although (\ref{eq:Sch C}) is holonomic
if we consider $\hbar$ a constant, 
(\ref{eq:Sch H}) is a PDE containing the
$\hbar$-differentiation as well. Therefore, 
it is not holonomic. This difference comes
from the constant term $2g-2+n$ in
the differential operator of 
the right-hand side of (\ref{eq:FgnH recursion}).
After taking the $n$-fold symmetric exterior
differentiation, this term drops, and thus 
the Eynard-Orantin recursion for Hurwitz
numbers \cite{EMS}
takes the same shape as that of
the Catalan case (\ref{eq:CEO}).

With the analysis of our examples, we
notice that the issue of the constants
of integration in (\ref{eq:W=dF}) is 
\emph{not} a simple matter. Only the
corresponding A-model can dictate which
constants of integration we should choose. 
Otherwise, the Schr\"odinger equation we
wish to establish would take a totally different
shape, depending on the choice of the constants.

Our second equation
(\ref{eq:Sch H2}) for Hurwitz numbers
is much similar to 
(\ref{eq:Sch C}) in many ways, such as
it is holonomic for each fixed $\hbar$. 
But this
differential-difference equation is not
a direct consequence of the 
differential equation (\ref{eq:FgnH recursion}),
while it recovers the Lagrangian immersion
more directly than
(\ref{eq:Sch H}). 
We also note here the commutator 
relation $[P,Q]=P$ of \cite{LMS} that we
mentioned in Introduction. We refer to
\cite{LMS} for a further integrable system
theoretic analysis of these equations.

Interestingly, we derive 
(\ref{eq:Sch H2}) from the fact that there
is another generating function for single 
Hurwitz numbers, which admits a Schur function
expansion. 
This last point is a more general feature. 
In \cite{MSS} we discover that there is 
a generalization of (\ref{eq:Sch H2}) for
the case of double Hurwitz numbers
and $r$-spin structures, which reduces
to (\ref{eq:Sch H2}) as a special case. 
A detailed analysis of double Hurwitz numbers,
or orbifold Hurwitz numbers, will be given elsewhere.

So far all these examples have a genus $0$
spectral curve. Hence the proposed algebraic 
K-theory obstruction of Gukov and
Su\l kowski \cite{GS} automatically vanishes.
It is also pointed out by Borot and 
Eynard \cite{BE2} that, for a 
higher genus spectral curve,
the definition of
the partition function of the B-model needs to 
be modified, by including a theta function 
factor known as a \emph{non-perturbative}
sector. This modification
 also assures the modular invariance of the
 partition function.
A further investigation is awaited here.

The examples we have carried
out in this paper suggest that the
expected Schr\"odinger equations for the
knot A-polynomials should form a rather 
special class
of the general Eynard-Orantin mechanism. 
This seems to be due to the integer-coefficient 
Laurent polynomial 
expression of the A-polynomials, and their 
$K_2$ Lagrangian property of Kontsevich.

\begin{appendix}
\section{Proof of the Schr\"odinger equation for 
the Catalan case}
\label{app:Catalan}

In this Appendix we give the proof of
Theorem~\ref{thm:Sch}. 
In Section~\ref{sect:CSch}
we reduced the proof to verifying 
the recursion formula (\ref{eq:S recursion}).
To prove this, 
we need the following  trivial lemma.

\begin{lem}
Let $f(t_1,\dots,t_n)$ be a symmetric function
in $n$ variables. 
Then
\begin{equation}
\begin{aligned}
\label{eq:dfdt}
\frac{d}{dt}f(t,t,\dots,t) 
&= n
\left.
\left[
\frac{\partial}{\partial u}f(u,t,\dots,t)
\right]
\right|_{u=t};
\\
\frac{d^2}{dt^2}f(t,t,\dots,t) 
&=
n\left.\left[
\frac{\partial^2}{\partial u ^2} 
f(u,t,\dots,t)
\right]\right|_{u=t}
\\
&\qquad
+
n(n-1)
\left.\left[
\frac{\partial^2}{\partial u_1 \partial u_2} 
f(u_1,u_2,t,\dots,t)
\right]\right|_{u_1=u_2=t}.
\end{aligned}
\end{equation}
For two functions in one variable 
$g(x)$ and $f(x)$,
we have
\begin{equation}
\label{eq:lhopital}
\left.
\left[
\frac{1}{x-y}\left(
g(x)\frac{df(x)}{dx}-g(y)\frac{df(y)}{dy}
\right)
\right]
\right|_{x=y}
=g'(x)f'(x)+g(x)f''(x).
\end{equation}
\end{lem}

\begin{proof}
For any function $f$ we have
$$
\frac{d}{dt}f(t,t,\dots,t) 
= 
\left.
\left[\sum_{j=1} ^n
\frac{\partial}{\partial t_j}f(t_1,t_2,\dots,t_n)
\right]
\right|_{t_1=t_2=\cdots=t_n=t}.
$$
Therefore,
\begin{multline*}
\frac{d^2}{dt^2}f(t,t,\dots,t) 
= 
\left.
\left[\left(\sum_{j=1} ^n
\frac{\partial}{\partial t_j}
\right)^2f(t_1,t_2,\dots,t_n)
\right]
\right|_{t_1=t_2=\cdots=t_n=t}
\\
=
\left.
\left[\left(\sum_{j=1} ^n
\frac{\partial^2}{\partial t_j^2}
+2\sum_{i\ne j}
\frac{\partial^2}{\partial t_i\partial t_j}
\right)
f(t_1,t_2,\dots,t_n)
\right]
\right|_{t_1=t_2=\cdots=t_n=t}.
\end{multline*}
Eq.(\ref{eq:dfdt}) holds when $f$ is symmetric.
The second equation (\ref{eq:lhopital})
follows from l'H\^opital's
rule.
\end{proof}

Now we are ready to give the proof of the 
recursion (\ref{eq:S recursion}).

\begin{proof}[Proof of {\rm{(\ref{eq:S recursion})}}]
The left-hand side of (\ref{eq:S recursion}) is
$$
\frac{d}{dt}S_{m+1}
=\sum_{2g-2+n=m}
\frac{1}{n!}\;\frac{d}{dt}F_{g,n}(t,\dots,t)
=\sum_{2g-2+n=m}\frac{1}{(n-1)!}
\left.
\left(
\frac{\partial}{\partial t_1}F_{g,n}(t_1,t,\dots,t)
\right)
\right|_{t_1=t}.
$$
Thus we apply the operation 
$$
\sum_{2g-2+n=m}\frac{1}{(n-1)!}
$$
to each line of the right-hand side of
(\ref{eq:FC recursion})
and set all variables equal to $t$.

From Line $1$ of the right-hand side of
(\ref{eq:FC recursion}), we have
\begin{multline*}
-\frac{1}{16}\sum_{2g-2+n=m}\frac{1}{(n-1)!}
\sum_{j=2} ^n
\frac{t_j}{(t_1+t_j)(t_1-t_j)}
\\
\times
\left.\left(
\frac{(t_1^2-1)^3}{t_1^2}\frac{\partial}{\partial t_1}
F^C_{g,n-1}(t_{[\hat{j}]})
-
\frac{(t_j^2-1)^3}{t_j^2}\frac{\partial}{\partial t_j}
F^C_{g,n-1}(t_{[\hat{1}]})
\right)
\right|_{t_1=\dots=t_n=t}
\\
=
-\frac{1}{32}\sum_{2g-2+n=m}\frac{1}{(n-2)!}
\Bigg[
\left(
\frac{d}{dt}\;\frac{(t^2-1)^3}{t^2}\right)
\frac{\partial}{\partial t_1}
F^C_{g,n-1}(t_1,t,\dots,t)
\\
+
\frac{(t^2-1)^3}{t^2}
\frac{\partial^2}{\partial t_1^2}
F^C_{g,n-1}(t_1,t,\dots,t)
\Bigg]\Bigg|_{t_1=t}
\\
=
-\frac{1}{32}
\left(
\frac{d}{dt}\;\frac{(t^2-1)^3}{t^2}\right)
\sum_{2g-2+n=m}\frac{1}{(n-1)!}
\frac{d}{dt}
F^C_{g,n-1}(t,\dots,t)
\\
-\frac{1}{32}
\frac{(t^2-1)^3}{t^2}
\sum_{2g-2+n=m}\frac{1}{(n-2)!}
\frac{\partial^2}{\partial t_1^2}
F^C_{g,n-1}(t_1,t,\dots,t)
\Bigg]\Bigg|_{t_1=t}
\\
=
-\frac{1}{32}
\left(
\frac{d}{dt}\;\frac{(t^2-1)^3}{t^2}\right)
\sum_{2g'-2+n'=m-1}\frac{1}{n'!}
\frac{d}{dt}
F^C_{g',n'}(t,\dots,t)
\\
-\frac{1}{32}
\frac{(t^2-1)^3}{t^2}
\sum_{2g'-2+n'=m-1}\frac{1}{(n'-1)!}
\frac{\partial^2}{\partial t_1^2}
F^C_{g',n'}(t_1,t,\dots,t)
\Bigg]\Bigg|_{t_1=t}
\\
=
-\frac{1}{32}
\left(
\frac{d}{dt}\;\frac{(t^2-1)^3}{t^2}\right)
\frac{d}{dt}S_m
\\
-\frac{1}{32}
\frac{(t^2-1)^3}{t^2}
\sum_{2g'-2+n'=m-1}\frac{1}{(n'-1)!}
\frac{\partial^2}{\partial t_1^2}
F^C_{g',n'}(t_1,t,\dots,t)
\Bigg]\Bigg|_{t_1=t}.
\end{multline*}
From Line $2$ we obtain
\begin{multline*}
-\frac{1}{16}\sum_{2g-2+n=m}\frac{1}{(n-1)!}
\sum_{j=2} ^n
\frac{(t_1^2-1)^2}{t_1^2}\frac{\partial}{\partial t_1}
F^C_{g,n-1}(t_{[\hat{j}]})\Bigg|_{t_1=\dots=t_n=t}
\\
=
-\frac{1}{16}
\frac{(t^2-1)^2}{t^2}\sum_{2g-2+n=m}
\frac{1}{(n-2)!}
\frac{\partial}{\partial t_1}
F^C_{g,n-1}(t_1,t,\dots,t)\Bigg|_{t_1=t}
\\
=
-\frac{1}{16}
\frac{(t^2-1)^2}{t^2}\frac{d}{dt}S_m.
\end{multline*}
Line $3$ produces
\begin{multline*}
-\frac{1}{32}\;
\sum_{2g-2+n=m}\frac{1}{(n-1)!}
\left[\frac{(t_1^2-1)^3}{t_1^2}
\left.
\frac{\partial^2}{\partial u_1\partial u_2}
F^C_{g-1,n+1}(u_1,u_2,t_2, t_3,\dots,t_n)
\right]
\right|_{u_1=u_2=t_1=\cdots=t_n=t}
\\
=
-\frac{1}{32}\;
\frac{(t^2-1)^3}{t^2}
\sum_{2g-2+n=m}\frac{1}{(n-1)!}
\left.
\left[
\frac{\partial^2}{\partial u_1\partial u_2}
F^C_{g-1,n+1}(u_1,u_2,t,\dots,t)
\right]
\right|_{u_1=u_2=t}
\\
=
-\frac{1}{32}\;
\frac{(t^2-1)^3}{t^2}
\sum_{2g'-2+n'=m-1, n'\ge 2}\frac{1}{(n'-2)!}
\left.
\left[
\frac{\partial^2}{\partial u_1\partial u_2}
F^C_{g',n'}(u_1,u_2,t,\dots,t)
\right]
\right|_{u_1=u_2=t}.
\end{multline*}
Finally, Line $4$ gives
\begin{multline*}
-\frac{1}{32}
\sum_{2g-2+n=m}\frac{1}{(n-1)!}
\Bigg[\frac{(t_1^2-1)^3}{t_1^2}
\\
\times
\sum_{\substack{g_1+g_2=g\\
I\sqcup J=\{2,3,\dots,n\}}}
^{\rm{stable}}
\left.
\frac{\partial}{\partial t_1}
F^C_{g_1,|I|+1}(t_1,t_I)
\frac{\partial}{\partial t_1}
F^C_{g_2,|J|+1}(t_1,t_J)\Bigg]
\right|_{t_1=\cdots=t_n=t}
\\
=
-\frac{1}{32}
\frac{(t^2-1)^3}{t^2}
\sum_{2g-2+n=m}\frac{1}{(n-1)!}
\\
\times
\sum_{\substack{g_1+g_2=g\\
n_1+n_2=n-1}}
^{\rm{stable}}
\left.
\Bigg[
\binom{n-1}{n_1}
\frac{\partial}{\partial t_1}
F^C_{g_1,n_1+1}(t_1,t,\dots,t)
\frac{\partial}{\partial t_1}
F^C_{g_2,n_2+1}(t_1,t,\dots,t)\Bigg]
\right|_{t_1=t}
\\
=
-\frac{1}{32}
\frac{(t^2-1)^3}{t^2}
\\
\times
\sum_{2g-2+n=m}
\sum_{\substack{g_1+g_2=g\\
n_1+n_2=n-1}}
^{\rm{stable}}
\left.
\Bigg[
\frac{1}{n_1!}
\frac{\partial}{\partial t_1}
F^C_{g_1,n_1+1}(t_1,t,\dots,t)
\frac{1}{n_2!}
\frac{\partial}{\partial t_1}
F^C_{g_2,n_2+1}(t_1,t,\dots,t)\Bigg]
\right|_{t_1=t}
\\
=
-\frac{1}{32}
\frac{(t^2-1)^3}{t^2}
\\
\times
\sum_{2g-2+n=m}
\sum_{\substack{g_1+g_2=g\\
n_1+n_2=n-1}}
^{\rm{stable}}
\frac{1}{(n_1+1)!}
\frac{d}{dt}
F^C_{g_1,n_1+1}(t,\dots,t)
\frac{1}{(n_2+1)!}
\frac{d}{dt}
F^C_{g_2,n_2+1}(t,\dots,t)
\\
=
-\frac{1}{32}
\frac{(t^2-1)^3}{t^2}
\sum_{\substack{a+b=m+1\\a,b\ge 2}}
\frac{dS_a}{dt}
\frac{dS_b}{dt}.
\end{multline*}
In the last line we note that unstable geometries
$(g,n) = (0,1)$ and $(0,2)$ are included only
in $S_0$ and $S_1$. 
This line gives the correct contribution of the
product term in (\ref{eq:S recursion}).

Since
$$
-\frac{1}{32}
\left(
\frac{d}{dt}\;\frac{(t^2-1)^3}{t^2}\right)
-\frac{1}{16}\;\frac{(t^2-1)^2}{t^3}
=
-\frac{1}{16}\;\frac{(t^2-1)^2}{t^3}(2t^2+t+1),
$$
we have the correct term for $dS_m/dt$ in 
(\ref{eq:S recursion}). The remaining
terms are second derivatives, and we calculate
\begin{multline*}
-\frac{1}{32}
\frac{(t^2-1)^3}{t^2}
\sum_{2g'-2+n'=m-1}\frac{1}{(n'-1)!}
\frac{\partial^2}{\partial t_1^2}
F^C_{g',n'}(t_1,t,\dots,t)
\Bigg|_{t_1=t}
\\
-\frac{1}{32}\;
\frac{(t^2-1)^3}{t^2}
\sum_{2g'-2+n'=m-1, n'\ge 2}\frac{1}{(n'-2)!}
\left.
\left[
\frac{\partial^2}{\partial u_1\partial u_2}
F^C_{g',n'}(u_1,u_2,t,\dots,t)
\right]
\right|_{u_1=u_2=t}
\\
=
-\frac{1}{32}
\frac{(t^2-1)^3}{t^2}
\sum_{2g'-2+n'=m-1}\frac{1}{n'!}
\\
\times
\left(
n'\frac{\partial^2}{\partial t_1^2}
F^C_{g',n'}(t_1,t,\dots,t)
\Bigg|_{t_1=t}
+
n'(n'-1)
\left.
\left[
\frac{\partial^2}{\partial u_1\partial u_2}
F^C_{g',n'}(u_1,u_2,t,\dots,t)
\right]
\right|_{u_1=u_2=t}
\right)
\\
=
-\frac{1}{32}
\frac{(t^2-1)^3}{t^2}
\frac{d^2}{dt^2}S_m.
\end{multline*}
This completes the derivation of
(\ref{eq:S recursion})
from (\ref{eq:FC recursion}).
\end{proof}

\section{Hierarchy of equations for $S_m$}

In this paper we proved that the Catalan and Hurwitz partition functions $Z$ satisfy appropriate Schr\"odinger equations, which can be written as
\begin{equation}
\widehat{A}\, Z = 0.  \label{AZ}
\end{equation}
In general such partition functions have an expansion $Z=\exp F = \exp\Big(\sum_{m=0}^{\infty} \hbar^{m-1} S_m\Big)$, where $S_m$ are expressed in terms of $F_{g,n}$ as in (\ref{eq:SmC})
\begin{equation}
S_m = \sum_{2g-2+n=m-1}\frac{1}{n!}\; F_{g,n},    \label{Sm-def}  
\end{equation}
while $\widehat{A}$ is a differential (or a difference-differential) operator expressed in terms of $x$ and $\hbar\frac{d}{dx}$. Our proofs relied on the knowledge of the recursion equations satisfied by $F_{g,n}$, and we did not need to determine coefficients $S_m$ explicitly. However, from the viewpoint of the asymptotic expansion in $\hbar$, these are the coefficients $S_m$ which play the fundamental role. They can be determined order by order in $\hbar$ once the form of the operator $\widehat{A}$ is known, or in turn -- if the form of $S_m$ is known, it allows to determine, also order by order in $\hbar$, the operator $\widehat{A}$ from the knowledge of its symbol $A=A(x,y)$. The relation between $S_m$ and $\widehat{A}$ can be encoded in a hierarchy of differential equations which was analyzed in \cite{GS}. In what follows we summarize the structure of this hierarchy and use it to determine several coefficients $S_m$ in the examples considered in this paper.

As $\widehat{A}$ is an operator expression, we should specify the ordering of $x$ and $\hbar\frac{d}{dx}$ operators it is built from. Let us choose the ordering such that, in each monomial summand, all $\hbar\frac{d}{dx}$ are to the right of $x$. In general, the Schr\"odinger operator can be written then in the form
\begin{equation}
\widehat{A} \; = \; \widehat{A}_0 + \hbar \widehat{A}_1 + \hbar^2 \widehat{A}_2 + \ldots \,,
\label{Ahatpert}
\end{equation}
which reduces in the $\hbar\to 0$ limit to the symbol $A_0=\widehat{A}_0=A=A(x,y)$. The examples considered in this paper are in fact quite special -- in both Catalan and Hurwitz case all $A_k=0$ for $k\geq 1$, and the issue of ordering is irrelevant (because monomials involving both $x$ and $\hbar\frac{d}{dx}$ do not arise), so that the Schr\"odinger operator can be obtained from the symbol $A_0$ just by the substitution $y\to\hbar \frac{d}{dx}$. Nonetheless, let us recall the most general form of the hierarchy of differential equations which relates a general operator of the form (\ref{Ahatpert}) to the coefficients $S_m$. This hierarchy can be written as \cite{GS}
\begin{equation}
\sum_{r=0}^n \D_{r} A_{n-r} \; = \; 0 \,,
\label{hierarchy-SA}    
\end{equation}
where $A_{n-r}$ are symbols of the operators $\widehat{A}_{n-r}$, and $\D_r$ are differential operators of degree $2r$, which can be
written as polynomials in $\partial_y \equiv \frac{\partial}{\partial y}$,
whose coefficients are polynomial expressions in functions $S_m$ and their derivatives (in what follows we denote $'\equiv\frac{d}{dx}$). 
The operators $\D_r$ are defined via the generating function
\begin{equation}
\sum_{r=0}^{\infty} \hbar^r \D_r \; = \;
\exp \left( \sum_{n=1}^{\infty} \hbar^n \frak{d}_n \right) \,, \nonumber
\label{Sigma-def}
\end{equation}
where
\be
\frak{d}_n \; = \; \sum_{r=1}^{n+1} \frac{S_{n+1-r}^{(r)}}{r!} (\partial_y)^r \,.  \nonumber
\ee
In particular, for small values of $n$ we get
\bea
\frak{d}_1 & = & \frac{1}{2} S''_0 \partial_y^2 + S'_1 \partial_y \,, \nonumber \\
\frak{d}_2 & = & \frac{1}{6}S'''_0 \partial_y^3 + \frac{1}{2} S''_1 \partial_y^2 + S'_2 \partial_y \,, \nonumber \\
\frak{d}_3 & = & \frac{1}{4!}S^{(4)}_0 \partial_y^4 +\frac{1}{3!}S'''_1 \partial_y^3 + \frac{1}{2} S''_2 \partial_y^2 + S'_3 \partial_y \,, \nonumber
\eea
so that 
\begin{subequations}\label{ddd}
\bea
\D_0 & = & 1 \,, \nonumber \\
\D_1 & = & \frac{S''_0}{2}  \partial_y^2 + S'_1 \partial_y \,,  \nonumber \\
\D_2 & = & \frac{(S''_0)^2}{8} \partial_y^4 +  \frac{1}{6}\big(S'''_0 + 3S''_0 S'_1 \big) \partial_y^3
+ \frac{1}{2}\big(S''_1 + (S'_1)^2 \big) \partial_y^2 + S'_2 \partial_y \,, \nonumber \\
& \vdots & \nonumber
\eea
\end{subequations}
In consequence, at each order $\hbar^n$ in \eqref{hierarchy-SA} we get:
\bea
\hbar^0 & : \qquad & A \; = \; 0 \,, \label{Aclass} \\
\hbar^1 & : \qquad & \Big(\frac{S''_0}{2}\partial_y^2 + S'_1 \partial_y \Big) A + A_1 \; = \; 0 \,, \label{hier-eq2v} \\
& & \quad \vdots \nonumber \\
\hbar^n & : \qquad & \D_n A + \D_{n-1} A_1 + \ldots + A_n \; = \; 0 \,, \label{hier-SA-Cstar} \\
& & \quad \vdots \nonumber
\eea

We can now use the above formalism to analyze Schr\"odinger equations for the generalized Catalan and Hurwitz numbers. We already stressed that these equations are special as they have no $\hbar$ corrections, i.e. all $A_k=0$ for $k\geq 1$. We can test this statement -- or, in other words, reconstruct entire $\widehat{A}$ from the form of $A$ -- using (\ref{hier-SA-Cstar}). In particular, from the Proposition \ref{prop:F01 F02}, for Catalan numbers we have
\be
S^C_0 = -\frac{1}{2}z^2 + \log z, \qquad \quad S^C_1 = -\frac{1}{2}\log (1-z^2).  \nonumber
\ee
Plugging into (\ref{hier-eq2v}) expressed in terms of $z$ and $\frac{d}{dx}=\big(\frac{dx}{dz}\big)^{-1}\frac{d}{dz}$ (and $x(z)=z+z^{-1}$), we indeed find $A_1=0$. From the knowledge of the recursion relations for $F_{g,n}$ we could now reconstruct other $S_m$ via (\ref{Sm-def}), and use the above hierarchy to show that higher $A_k$ vanish as well. However, from our earlier considerations we have essentially proven this, using the form of recursion relations for $F_{g,n}$, without the need of writing down $S_m$ explicitly. Nonetheless, in various situations it is necessary to know the form of coefficients $S_m$ or the recursion they satisfy, so we can also use the above formalism in this context. The equation (\ref{hier-SA-Cstar}) indeed gives a recursion directly for $S_m$, which in addition (due to $A_k=0$ for $k\geq 1$) reduces simply to $\D_n A=0$. Expressing this recursion in terms of $z$-variable, we find in particular
\bea
\frac{dS^C_2}{dx} &=& \frac{z^5 (2 z^2+3)}{(z^2-1)^5}, \nonumber \\
\frac{dS^C_3}{dx} &=& -\frac{5 z^7 (3 + 7 z^2 + 2 z^4)}{(z^2-1)^8}, \nonumber \\
\frac{dS^C_4}{dx} &=& \frac{1 - 9 z^2 + 36 z^4 - 84 z^6 - 4599 z^8 - 13005 z^10 - 4440 z^{12}}{360 (-1 + z^2)^9}. \nonumber
\eea
Moreover, from (\ref{eq:FgnC zero}) we have $S^C_m(z=0)\, = \, 0$, which fixes a constant of integration, so that integrating the above formulas gives
\bea
S^C_2 &=& \frac{z^4 (9 + z^2)}{12 (1 - z^2)^3}, \nonumber \\
S^C_3 &=& \frac{5 z^6(1 + z^2)}{2 (z^2)^6 - 1}, \nonumber \\
S^C_4 &=& \frac{z^8 (-4725 - 12879 z^2 - 4524 z^4 + 36 z^6 - 9 z^8 + z^{10})}{360 (z^2-1)^9}. \nonumber
\eea

This exercise can also be repeated for the case of the generalized Hurwitz numbers. From (\ref{eq:F01H}) and (\ref{eq:F02H diag}) we have
\be
S^H_0 = z-\frac{1}{2}z^2, \qquad \quad S^H_1 = -\frac{1}{2}\Big(z + \log (1-z)\Big) .    \nonumber
\ee
Plugging into (\ref{hier-eq2v}) confirms that the first correction to the classical Hurwitz curve vanishes, $A_1=0$, as it should. Solving further the hierarchy (\ref{hier-SA-Cstar}), which again reduces to $\D_n A=0$, results in
\bea
\frac{dS^H_2}{dx} &=& \frac{z^3 (4 + z^2)}{8 (1 - z)^5}, \nonumber \\
\frac{dS^H_3}{dx} &=& -\frac{z^4(12 + 8 z + 9 z^2 + z^4)}{16 (z-1)^8}, \nonumber \\
\frac{dS^H_4}{dx} &=& \frac{z^5(192 + 352 z + 376 z^2 + 104 z^3 + 76 z^4 + 5 z^6)}{128 (1 - z)^{11}}. \nonumber
\eea
These equations can be further integrated to get $S^H_m$, with the integration constant fixed by $S^H_m(z=0)\,=\,0$, which follows from (\ref{eq:FgnH zero}).

\bigskip


\end{appendix}

\begin{ack}
The authors thank the Banff International 
Research Center in Alberta and the Hausdorff 
Research Institute for 
Mathematics in Bonn for their
support and hospitality, where this collaboration
was started. They also thank 
Ga\"etan Borot,
Vincent Bouchard, 
Bertrand Eynard,
 Sergei Gukov,
  Jerry Kaminker,
   Maxim Kontsevich,
 Xiaojun Liu, 
 Marcos Mari\~no,
 Michael Penkava,
Anne Schilling, 
Sergey Shadrin, 
and
Don Zagier
for useful discussions. 
The research of M.M.\ has been 
supported by NSF DMS-1104734, DMS-1104751,
Max-Planck Institut f\"ur Mathematik in Bonn,
 the Beijing International Center
for Mathematical Research, 
American Institute of Mathematics in Palo Alto,
the University of Salamanca, 
and  Universiteit van Amsterdam. 
The research 
of P.S.\ has been supported by 
the DOE grant DE-FG03-92-ER40701FG-02, 
the European Commission under the Marie-Curie 
International Outgoing Fellowship Programme, 
and the Foundation for Polish Science.
\end{ack}


\providecommand{\bysame}{\leavevmode\hbox to3em{\hrulefill}\thinspace}

\bibliographystyle{amsplain}

\end{document}